\documentclass[12pt, a4paper]{article}
\usepackage[utf8]{inputenc}
\usepackage[english]{babel}
\usepackage{amsmath}
\usepackage{amsthm}

\usepackage{dsfont}

\usepackage{soul} 
\usepackage{cancel}

\usepackage{hyperref}

\usepackage{lmodern}
\usepackage[a4paper, top=3cm, bottom=4cm, left=3cm, right=3cm]{geometry}
\usepackage{xcolor}
\newlength{\headlineindent}
\newcommand{\equationindent}{\quad \quad}
\newcommand{\equationspacing}{\quad}

\title{On the Ultraviolet Problem for the Ground State Energy of the
  Translation-Invariant Pauli--Fierz Model at Zero Total Momentum }
\author{Volker Bach, Miguel Ballesteros, Merten Mlinarzik}
\date{24-Oct-2025}

\definecolor{gr}{rgb}{0.2,0.6,0.2}

\usepackage{amsmath}
\usepackage{mathtools}
\usepackage{amsfonts}
\usepackage{amssymb}
\usepackage{amsthm}

\numberwithin{equation}{section}

\theoremstyle{plain}
\newtheorem{thm}{Theorem}[section]
\newtheorem{lem}[thm]{Lemma}
\newtheorem{prop}[thm]{Proposition}

\theoremstyle{definition}
\newtheorem{defn}[thm]{Definition}
\theoremstyle{remark}
\newtheorem{rem}[thm]{Remark}

\makeatletter

\DeclarePairedDelimiterX{\@of}[1]{(}{)}{#1}
\newcommand{\of}{\@of*}

\DeclarePairedDelimiterX{\@norm}[1]{\lVert}{\rVert}{#1}
\newcommand{\norm}{\@norm*}

\DeclarePairedDelimiterX{\@abs}[1]{\lvert}{\rvert}{#1}
\newcommand{\abs}{\@abs*}

\DeclarePairedDelimiterX{\@brak}[1]{[}{]}{#1}
\newcommand{\brak}{\@brak*}

\DeclarePairedDelimiterX{\@ket}[1]{\vert}{\rangle}{#1}
\newcommand{\ket}{\@ket*}

\DeclarePairedDelimiterX{\@bra}[1]{\langle}{\vert}{#1}
\newcommand{\bra}{\@bra*}

\DeclarePairedDelimiterX{\@sset}[1]{\{}{\}}{#1}
\newcommand{\sset}{\@sset*}

\DeclarePairedDelimiterX{\@set}[2]{\{}{\}}{#1 \, \delimsize\vert \, #2}
\newcommand{\set}{\@set*}

\DeclarePairedDelimiterX{\@sca}[2]{\langle}{\rangle}{#1 \, \delimsize\vert \, #2}
\newcommand{\sca}{\@sca*}

\makeatletter
              
\newcommand{\too}[2][]{\xrightarrow[#1]{#2}}

\newcommand{\eps}{{\varepsilon}}  
\newcommand{\vphi}{{\varphi}}     
\newcommand{\Om}{\Omega}          
\newcommand{\om}{\omega}
\newcommand{\la}{\langle}
\newcommand{\ra}{\rangle}

\newcommand{\bfone}{\mathbf{1}}

\newcommand{\cB}{\mathcal{B}}
\newcommand{\cD}{\mathcal{D}}
\newcommand{\cE}{\mathcal{E}}
\newcommand{\tcE}{\widetilde{\mathcal{E}}}

\newcommand{\cG}{\mathcal{G}}
\newcommand{\tcG}{\widetilde{\mathcal{G}}}
\newcommand{\cH}{\mathcal{H}}

\newcommand{\cL}{\mathcal{L}}         
         
\newcommand{\cN}{\mathcal{N}}         
\newcommand{\cO}{\mathcal{O}}

\newcommand{\field}[1]{\mathds{#1}}
\newcommand{\bbA}{\field{A}}

\newcommand{\RR}{\field{R}}     
     
\newcommand{\NN}{\field{N}}     
\newcommand{\CC}{\field{C}}     
\newcommand{\bbS}{\field{S}}     
     
\newcommand{\UU}{\field{U}}     
\newcommand{\bbW}{\field{W}}

\newcommand{\fF}{\mathfrak{F}}
\newcommand{\fH}{\mathfrak{H}}

\newcommand{\fh}{\mathfrak{h}}

\newcommand{\sfJ}{\mathsf{J}}

\newcommand{\uA}{{\underline A}}

\newcommand{\uRR}{{\underline \RR}}

\newcommand{\uS}{{\underline S}}

\newcommand{\uk}{{\underline k}}

\newcommand{\rRe}{\mathrm{Re}}

\newcommand{\spec}{\sigma}

\newcommand{\cirS}{\mathop{\bigcirc\kern -.73em {\scriptstyle{\rm S}}}}

\newcommand{\huelle}{{\rm span}}

\newcommand{\Tr}{{\rm Tr}}

\newcommand{\op}{\mathrm{op}} 
\newcommand{\HS}{\mathrm{HS}}
\newcommand{\Bog}{\mathrm{Bog}}
\newcommand{\el}{\mathrm{el}}
\newcommand{\ph}{\mathrm{ph}}
\newcommand{\gs}{\mathrm{gs}}
\newcommand{\prtcl}{\mathrm{part}}
\newcommand{\BHF}{\mathrm{BHF}}

\newcommand{\DM}{\mathcal{D}\!\mathcal{M}}

\newcommand{\pQDM}{\mathcal{P}\!\mathcal{Q}\!\mathcal{D}\!\mathcal{M}}

\newcommand{\LL}{{\mathrm{LL}}}

\begin{document}

\maketitle

\begin{abstract}
We study the ground state energy of the Pauli--Fierz model in the
absence of external potentials. We consider the fiber decomposition of
the Pauli--Fierz operator with respect to the spectral values, $p$, of
the total momentum operator and focus on the case $p = 0$. The
corresponding variational problem is analyzed to estimate the
dependence of the ground state energy on the ultraviolet cutoff
$\Lambda$. We employ a Bogoliubov--Hartree--Fock approximation using
pure, quasifree states generated by Bogolubov transformations
(parametrized by a positive Hilbert--Schmidt operator $z$) and Weyl
transformations (parametrized by a vector $\eta$) applied to the
vacuum. We prove that the resulting energy functional is not a convex
function of $\eta$ and $z$. We identify the non-convex term and remove
it from the energy functional. The modified functional retains the
full interaction term and is shown to be strictly convex. We study the
ground state of the modified functional and prove the existence of a
unique minimizer. Furthermore, we construct an explicit partial
minimizer (with respect to $\eta$, for fixed $z$), which allows us to
eliminate $z$ and reduce the minimization problem to a single
variable, $\eta$. Finally, we estimate the minimum of the modified
energy functional in terms of the ultraviolet cutoff $\Lambda$ and
demonstrate that, up to a constant factor, it grows asymptotically as
$\Lambda^{3/2}$, as $\Lambda \to \infty$.
\end{abstract}

\section{Introduction} 
\label{sec:Introduction}

In this article, we study spectral properties of the Pauli--Fierz
Hamiltonian generating the dynamcs of a single nonrelativistic
electron interacting with the quantized radiation field in absence of
external potentials, with a specific focus on the ultraviolet
limit. We assume an ultraviolet cutoff $\Lambda > 0$ and study the
Pauli--Fierz operator $H^{g,\Lambda}$ (where $g$ is the coupling
constant) for zero total momentum. Since $H^{g,\Lambda}$ commutes with
the total momentum operator (due to translation invariance), it can be
decomposed into a direct integral of fiber Hamiltonians
$(H_p^{g,\Lambda})$ with respect to the total momentum value 
$p \in \RR^3$. This fiber decomposition eliminates the electron
variable, and each operator $H_p^{g,\Lambda}$ acts on the photon Fock
space. In this article, we focus on the \textit{ground state energy}
$E_\gs(H_0^{g,\Lambda})$ of $H_0^{g,\Lambda}$ and study the
corresponding variational problem,

\begin{align}\label{eq:Intro_VariationalProblem-1}
E_\gs(H_0^{g,\Lambda}) 
\ := \  
\inf\Big\{ \Tr\big[ \rho \, H_0^{g,\Lambda} \big] \ \Big|
\ \rho \; \in \; \DM(\fF_\ph) \, , \ 
(\rho H_0^{g,\Lambda}), (H_0^{g,\Lambda} \rho) 
\: \in \: \cL^1(\fF_\ph) \Big\} \, , 
\end{align}
where $\DM(\fF_\ph) := \{ \rho \in \cL^1(\fF_\ph) \, | \; \rho \geq 0,
\; \Tr(\rho) = 1 \}$ are the density matrices on photon Fock space
$\fF_\ph$. The density matrices $\DM(\fF_\ph) \subseteq \cL^1(\fF)$
form a convex subset of the space $\cL^1(\fF)$ of trace-class
operators on $\fF_\ph$.

The minimization problem \eqref{eq:Intro_VariationalProblem-1} is
highly non-trivial, necessitating simplification to extract meaningful
results. In this paper, we restrict the variation in
\eqref{eq:Intro_VariationalProblem-1} to those density matrices 
$\rho \in \DM(\fF_\ph)$ which are \textit{quasifree}. The
corresponding variational problem
\begin{align}\label{eq:Intro_VariationalProblem-2}
& E_\BHF(H_0^{g,\Lambda}) 
\ := \  \\ \nonumber
& \inf\Big\{ \Tr\big[ \rho \, H_0^{g,\Lambda} \big] \ \Big|
\ \rho \; \in \; \DM(\fF_\ph) \, , \ 
(\rho H_0^{g,\Lambda}), (H_0^{g,\Lambda} \rho) 
\: \in \: \cL^1(\fF_\ph) \, , \ \text{$\rho$ is quasifree} \Big\} 
\end{align}

defines the \textit{Bogoliubov--Hartree--Fock energy}
$E_\BHF(H_0^{g,\Lambda}) \geq E_\gs(H_0^{g,\Lambda})$ of the
Pauli--Fierz operator for zero total momentum, which is an upper bound
for the ground state energy. It was shown in
\cite{DerezinskiNapiorkowskiSolovej2013, BachBreteauxTzaneteas2013}
that the infimum in \eqref{eq:Intro_VariationalProblem-2} is already
obtained by a minimization over \textit{pure} states.  These states
are rank-one orthogonal projections onto vectors obtained by applying
(homogeneous) Bogoliubov transformations followed by Weyl
transformations to the vacuum state.

More specifically, let $\fh_\ph^\Lambda$ denote the one-photon Hilbert
space and 
$\Bog(\fh_\ph^\Lambda) \subseteq \cB[\fh_\ph^\Lambda \oplus
    \fh_\ph^\Lambda]$ denote the set of Bogoliubov matrices on
$\fh_\ph^\Lambda$. For every $B \in \Bog(\fh_\ph^\Lambda)$, we denote
by $\UU_B$ the corresponding unitary Bogoliubov transformation
[defined in \eqref{eq:Intro_BogAction} below].  Additionally, we
denote by $\bbW_\eta$ the Weyl transformation associated with 
$\eta \in \fh_\ph^\Lambda$ [see \eqref{eq:Intro_WeylAction}]. The
Bogoliubov--Hartree--Fock energy can now be written as

\begin{align}\label{eq:Intro_BHFEnergyvalue}
E_\BHF(H_0^{g,\Lambda}) 
\ = \ 
\inf\Big\{ \big\la \bbW_\eta \UU_B \Om \, , \; 
H_0^{g,\Lambda} \, \bbW_\eta \UU_B \Om \big\ra \; \Big| \ 
B \in \Bog(\fh_\ph^\Lambda) \, , \ \eta \in \fh_\ph^\Lambda \Big\} \, , 
\end{align}

where the variation is restricted to $B$ and $\eta$ such that
$(H_0^{g,\Lambda})^{1/2} \bbW_\eta \UU_B \Om \in \fF_\ph$.  [see
  \eqref{eq:Intro_DefinitionBHFEnergy}]. An explicit formula for
$E_\BHF(H_0^{g,\Lambda})$ was derived in
\cite{BachBreteauxTzaneteas2013} and considerably simplified in
\cite{Herdzik2023} [see \eqref{eq:Intro_RawEnergyFunctional}], where
it was shown that the Bogoliubov matrices $B$ over which the energy
expectation value $\la \bbW_\eta \UU_B \Om , \, H_0^{g,\Lambda}
\bbW_\eta \UU_B \Om \ra$ is varied, can be parametrized by a single,
positive Hilbert--Schmidt operator $z$, i.e., $B \equiv B(z)$. It was
further noticed in \cite{Herdzik2023} that $E_\BHF(H_0^{g,\Lambda})$
is bounded above by the infimum of the energy functional

\begin{align}\label{eq:Intro_BHFEnergyFunctional}
\cE_{g,\Lambda}(z, \eta)
\ := \ 
\big\la \bbW_\eta \UU_{B(z)} \Om \, , \; 
H_0^{g,\Lambda} \, \bbW_\eta \UU_{B(z)} \Om \big\ra \, , 
\end{align}

minimized only over the subset of \textit{$\sfJ$-invariant} Bogoliubov
transformations $\UU_{B(z)}$ and Weyl transformations $\bbW_\eta$,
with $z = \sfJ z \sfJ$ and $\eta = \sfJ \eta$. Moreover, if the
infimum of $\cE_{g,\Lambda}(z, \eta)$ is attained for a unique
$\sfJ$-invariant minimizer, then it even equals
$E_\BHF(H_0^{g,\Lambda})$. This result from \cite{Herdzik2023} is the
starting point of the present paper.

In this article, we analyze the variational problem of minimizing
$\cE_{g,\Lambda}$ and derive estimates in terms of the ultraviolet
cutoff $\Lambda \gg 1$ and coupling constant $g \ll 1$. Our first
result establishes that $\cE_{g,\Lambda}$ is not convex, identifying
the reason for this lack of convexity to be a term arising from the
square of the total momentum operator
(Theorem~\ref{thm:nonConvexity}). We introduce a simplified
variational problem by a reduced functional $\tcG_{g,\Lambda}$ in
which we retain the full interaction while omitting this non-convex
term (Definition~\ref{def:Intro_GTilde}).

To elucidate the $\Lambda$-dependence of $\tcG_{g,\Lambda}$, we change
the physical units by means of a suitable scaling transformation that
normalizes the energy $\Lambda$ to unity. The rescaled functional
(Definition~\ref{def:Intro_G}), denoted $\cG_{g,\Lambda}$, is the
central object of study. We address the minimization problem and
estimate the minimum in terms of $\Lambda$. Theorem
\ref{thm:existence} establishes the existence of a unique
minimizer. In Theorem~\ref{thm:reducedFunctional}, we demonstrate
that, given a fixed Bogoliubov parameter $z$, the reduced functional
$\cG_{g,\Lambda}$ admits an explicit unique partial minimizer
$\eta_z$, parametrizing the Weyl transformations. Substituting this
minimizer into the energy functional yields a reduced variational
problem depending solely on the Bogoliubov parameter
(Theorem~\ref{thm:reducedFunctional}). Our main result,
Theorem~\ref{thm:groundStateEnergyEstimate}, states that, to leading
order in $g$ and $\Lambda$, the minimum $E_\cG^{g,\Lambda}$ of
$\cG_{g,\Lambda}$ satisfies the bounds:

\begin{align} \label{eq-lead-ord-bd_1}
4\sqrt{\pi/3} \, g \Lambda^{3/2} 
\ \leq \ 
E_\cG^{g,\Lambda} 
\ \leq \ 
4\sqrt{3\pi} \, g \Lambda^{3/2} \, .
\end{align}

The first pillar on which the present paper rests is the important
work \cite{LiebLoss2000} in which a new variational model yielding an
upper bound $E_\LL^{g,\Lambda}$, called the \textit{Lieb--Loss energy}
in \cite{BachHach2022}, on the ground state energy
$E_\gs(H^{g,\Lambda})$ of a single nonrelativistic electron in absence
of external potentials and interacting with the quantized radiation
field had been introduced and analyzed. The functional defining the
Lieb--Loss energy derives from the variational characterization
$E_\gs(H^{g,\Lambda}) = \inf_{\|\Psi\| = 1} \la \Psi , \,
H^{g,\Lambda} \Psi \ra$ by restricting the variation to wave functions
of the form $\Psi = \vphi_\el \otimes \psi_\ph$, where $\vphi_\el$ is
a normalized electron wave function and $\psi_\ph$ a normalized vector
in photon Fock space $\fF_\ph$.  For the first time, non-trivial
estimates on the asympotics of the ground state energy in the
ultraviolet limit $\Lambda \to \infty$ were derived in
\cite{LiebLoss2000}, namely,

\begin{align}  \label{eq-LL-asymptotics.1}
C_1 \, \alpha^{1/2} \, \Lambda^{3/2}
\ \leq \ 
E_\gs(H^{g,\Lambda})
\ \leq \ 
E_\LL^{g,\Lambda}
\ \leq \ 
C_2 \, \alpha^{2/7} \, \Lambda^{12/7} \, ,
\end{align}  

for some constants $C_1, C_2 > 0$. Note that the lower bound results
from an estimate on the quadratic form defined by $H^{g,\Lambda}$ and
not from a variational bound. It was conjectured in
\cite{LiebLoss2000}, that the upper bound in
\eqref{eq-LL-asymptotics.1} is actually the correct asympotics of
$E_\gs(H^{g,\Lambda})$. The analysis of the Lieb-Loss model was
continued in \cite{BachHach2022}, and there it was shown that this
conjecture indeed holds true for the Lieb--Loss energy, namely, that
$E_\LL^{g,\Lambda} \sim C_\LL \alpha^{2/7} \Lambda^{12/7}$,

\begin{align}  \label{eq-LL-asymptotics.2}
E_\LL^{g,\Lambda}
\ \sim \ 
F_1 \, \alpha^{2/7} \, \Lambda^{12/7} \, ,
\end{align}  

for some explicitly given numerical constant $F_1 >0$. Further
developments in \cite{BreteauxFaupinPayet2023,
  BreteauxFaupinPayet2025} explore the Lieb--Loss variational problem
for coherent photon states $\psi_\ph = \bbW_\eta \Om$ in the cases of
the Pauli--Fierz operator with linear coupling
\cite{BreteauxFaupinPayet2023} and minimal coupling
\cite{BreteauxFaupinPayet2025}.

An equally important pillar for the present paper is the
Bogoliubov--Hartree--Fock theory that was defined and analyzed in
\cite{BachBreteauxTzaneteas2013}. Like the Lieb--Loss energy, the
Bogoliu\-bov--Hartree--Fock energy $E_\BHF(H_p^{g,\Lambda})$ is defined
by a variational principle and constitutes an upper bound on the
ground state energy of the Pauli--Fierz Hamiltonian, but for fixed
total momentum $p$. In fact, the Lieb-Loss model and the
Bogoliubov--Hartree--Fock theory are closely related, as they both
lead to a variation over quasi-free photon states. The investigations
in \cite{BachBreteauxTzaneteas2013} and several important observations
borrowed from \cite{BachHach2022} lead to the recent work
\cite{Herdzik2023}, which is the starting point of the present paper.

Estimating the ground state energy for the Pauli--Fierz model remains
a longstanding challenge, with numerous results addressing this
problem in the literature.

The infrared divergence in translation-invariant models was first
investigated in the 1970s in \cite{Froehlich1973}. More recent work
in \cite{HaslerHerbst2008a} established the absence of ground states
in such translation-invariant systems. Early progress on massive field
couplings appeared in \cite{Froehlich1974}, where the existence of a
ground state for an electron interacting with a massive bosonic field
was proven under fixed total momentum conditions. Subsequent
developments in \cite{Arai1983, Spohn1997} demonstrated ground state
existence for systems with $x^2$-type external potentials and
perturbations.

A significant breakthrough emerged with the introduction of spectral
renormalization techniques in \cite{BachFroehlichSigal1995,
  BachFroehlichSigal1998a, BachFroehlichSigal1998b}, which enabled
rigorous analysis of spectral properties including ground states and
resonances for infrared-regularized models with ultraviolet cutoffs.
These methods were later extended to non-regularized models in
\cite{BachFroehlichSigal1999} and further improved with optimal
conditions in \cite{GriesemerLiebLoss2001}. Complementary approaches
were developed in \cite{BachFroehlichPizzo2006,
  BarbarouxChenVugalter2003, HaslerHerbst2011, HaslerSiebert2024} both
for confined and for translation-invariant systems. Alternative
treatments under more restrictive hypotheses were presented
in~\cite{Hiroshima2000b}, while related models were investigated
in~\cite{AraiHirokawa2000, Gerard2000, DerezinskiGerard1999}.

The multiscale analysis framework was introduced in \cite{Pizzo2003}
for ground state studies. The complex multi-scale analysis suitable
for resonances was developed in \cite{BachBallesterosPizzo2017}. The
method of \cite{BachBallesterosPizzo2017} has proven equally powerful
as spectral renormalization for quantum electrodynamics problems.

Resonances for non-regularized infrared problems were proven to exist
in \cite{Sigal2009} (spectral renormalization) and in
\cite{BachBallesterosIniestaPizzo2021} (multi-scale analysis).

The construction of the ultraviolet limit has been successfully
carried out for some models in quantum field theory: the Nelson model
by a Gross transformation~\cite{Gross1962,Nelson1964}, the spin-boson
Hamiltonian~\cite{DamMoeller2020}, the Fr\"ohlich Hamiltonian
\cite{GriesemerLinden2019, Lampart2020}, and the Yukawa model
\cite{DeckertPizzo2014}.

\subsection{Organisation of the paper} 
\label{subsec:Organisation}

Subsection~\ref{subsec:Framework} provides an overview of the
mathematical framework of the paper and the results of
\cite{Herdzik2023}. In particular, the Hamiltonian of the Pauli--Fierz
modell as well as its Bogoliubov--Hartree--Fock energy is
introduced. The derivation of the full energy functional
$\cE_{g,\Lambda}$, as obtained in \cite{BachBreteauxTzaneteas2013} and
in \cite{Herdzik2023}, is briefly sketched. In
Subsection~\ref{subsec:MainResults} we present the main results and
central definitions of the paper in detail.

The introductory section is followed by the main body of the paper in
Sections~\ref{sec:Properties}-\ref{sec:NonConvexity}. In
Section~\ref{sec:Properties} we establish some fundamental variational
properties of the reduced energy functional $\cG_{g,\Lambda}$,
including convexity and coercivity. These properties are then used in
Section~\ref{sec:Existence} to show the existence of a unique
minimizer of $\cG_{g,\Lambda}$. In
Section~\ref{sec:PartialMinimization} we consider the partial
minimization problems of $\cG_{g,\Lambda}$. This is followed by
Section~\ref{sec:GroundStateEnergy} where we prove explicit estimates
for the minimal energy of $\cG_{g,\Lambda}$. To obtain these estimates
some key operator inequalities are used which are discussed in detail
in Section~\ref{sec:KeyInequalities}. In
Section~\ref{sec:NonConvexity} we provide a proof for the asserted
lack of convexity of the full energy functional $\cE_{g,\Lambda}$ by
identifying a class of counterexamples in the space of rank-two
operators.

For the convenience of the reader, the Appendix contains the proofs of
some general well-known results, which are used throughout this paper.

\subsection{Mathematical Framework and Previous Results}
\label{subsec:Framework}

This section is dedicated to introducing the mathematical framework of
the model and related notations. We also include a brief overview of
previous results.

\paragraph*{Notation.} 
We begin by presenting some notation which is used throughout this
work. We write $\uk$ for an element of the set

\begin{align} \label{eq-001}
\uRR^3 
\ := \ 
\RR^3 \times \sset{+,-} 
\ = \
\set{\uk = \of{k,\tau}}{k \in \RR^3 , \tau \in \sset{+,-}} \, .
\end{align}

For any vector $\uk = \of{k,\tau} \in \uRR^3$ and $\alpha \in \RR$ we
adopt the convention

\begin{align} \label{eq-002}
\alpha \uk 
\ := \ 
\of{\alpha k, \tau} \, .
\end{align}

Given a subset $A \subset \RR^3$ of $\RR^3$, we introduce the set 
$\uA \subset \uRR^3$ by

\begin{align} \label{eq-003}
\uA \ := \ 
\set{\uk = \of{k, \tau}}{k \in A, \tau \in \sset{+,-}} \, . 
\end{align}

For any integrable  function $f: \uA \to \CC$ we denote by 

\begin{align} \label{eq-004}
\int_{\uA} f\of{\uk} \, d \uk 
\ := \
\sum_{\tau \in \sset{+,-}} \int_A f\of{k,\tau} \, dk \, .
\end{align}

Furthermore, for two positive numbers $0 < a < b$ we introduce the subset 

\begin{align} \label{eq:Intro_DefinitionOfS}
\uS\of{a,b} 
\ := \
\set{\uk = \of{k,\tau} \in \uRR^3}{a \leq \abs{k} \leq b} \subset \uRR^3
\, .
\end{align}

For any complex and separable Hilbert space $\fh$ we denote by 

\begin{align} \label{eq-005}
\cB\of{\fh} 
\ := \
\set{z : \fh \to \fh}{z \textnormal{ is linear and bounded}}
\end{align}

the Banach space of bounded linear operators on $\fh$. Additionally,
we denote by

\begin{align} \label{eq-006}
\HS\of{\fh} 
\ := \
\set{z \in \cB\of{\fh}}{\Tr\brak{z^* z} < \infty}
\end{align}

the set of Hilbert--Schmidt operators on $\fh$, which
constitutes a Hilbert space when equipped with the usual
Hilbert--Schmidt norm

\begin{align} \label{eq-007}
\forall z \in \HS\of{\fh}: 
\qquad
\norm{z}_\HS 
\ := \
\sqrt{ \Tr\brak{z^* z} \, } \, .
\end{align}

Furthermore, for any $\delta \in \RR$, we define the sets

\begin{align} \label{eq:Intro_DefinitionHilbertSchmidtSet}
\HS_{\delta}\of{\fh} 
\ := \
\set{z \in \HS\of{\fh}}{ z=z^*, \ z \geq -\delta } \, .
\end{align}

Note that $\HS_\delta\of{\fh}$ is a convex subset of $\HS\of{\fh}$,
for any $\delta \in \RR$.  $\HS_{0}\of{\fh}$ is the set of positive
(self-adjoint) Hilbert--Schmidt operators on $\fh$.

\paragraph*{The Pauli--Fierz Model.} We study the interaction of a
single non-relativistic and charged particle with the quantized
electro-magnetic radiation field. The particle is assumed to be
spinless.  

Let $x \in \RR^3$ denote the position of the particle
in physical space. In the absence of an external potential the free
motion of the particle with mass $m > 0$ is modeled by the
Schr\"odinger Hamiltonian

\begin{align} \label{eq-008}
H_\prtcl \ := \ 
\frac{1}{2m}\of{i\vec{\nabla}_x}^2 \, ,
\end{align}

acting on the Hilbert space 

\begin{align} \label{eq-009}
\fH_{\prtcl} \ := \ L^2\of{\RR_x^3;\CC} \, .
\end{align}

The particle Hamiltonian is self-adjoint on its natural domain

\begin{align} \label{eq-010}
D\of{H_{\prtcl}} \ = \ H^2\of{\RR^3_x;\CC}
\, ,
\end{align}

the Sobolev space of square-integrable functions with
square-integrable distributional derivatives up to second order
\cite[Section~IX.7]{ReedSimonII1980}.

The radiation field are the photons. The single-photon space of
photons with energies between $\sigma >0$ and $\Lambda > \sigma$ is
denoted by

\begin{align} \label{eq-011}
\fh_{\ph}^\Lambda
\ := \ 
L^2\of{\uS\of{\sigma,\Lambda};\CC} \, ,
\end{align}

where $\uS\of{\sigma,\Lambda} \subset \uRR^3_{k}$, as
defined in \eqref{eq:Intro_DefinitionOfS}. The positive numbers
$\sigma$ and $\Lambda$ represent an infrared and ultraviolet cut-off
for the photon energies, respectively. While imposing such cut-offs is
unphysical it is nevertheless necessary for the mathematical
construction of the model. As this work is primarily concerned with
the ultraviolet divergence of the Pauli--Fierz model in the limit
$\Lambda \to \infty$, the infrared cut-off $\sigma$ remains fixed
throughout the work. The time evolution of a single free photon is
governed by the dispersion relation

\begin{align} \label{eq-012}
\om\of{\uk} 
\ := \ 
\abs{k} \, ,
\end{align} 

acting as a bounded multiplication operator on $\fh_\ph^\Lambda$.

The Hilbert space of the quantized radiation field
$\fH^\Lambda_{\ph}$ is given by the bosonic Fock space
$\fF\big[ \fh_\ph^\Lambda \big]$ associated to the
single particle photon space $\fh_\ph^\Lambda$. More precisely
we define

\begin{align} \label{eq-013}
\fH_{\ph}^\Lambda 
\ := \ 
\fF\big[ \fh_\ph^\Lambda \big] 
\ := \
\bigoplus_{n \geq 0} \fF^{\of{n}}\big[ \fh_\ph^\Lambda \big] 
\, , \quad 
\fF^{\of{0}}\big[ \fh_\ph^\Lambda \big] 
\ := \ \CC 
\, , \quad 
\fF^{\of{n}}\big[ \fh_\ph^\Lambda \big] 
\ := \ \big( \fh_\ph^\Lambda \big)^{n \odot} \, ,
\end{align}

where $n \odot$ denotes the $n$-fold symmetric tensor product of a
Hilbert space. We refer to the Fock vector

\begin{align} \label{eq-014}
\Om \ := \ 
\of{1,0,0,\ldots} 
\ \in \ \fH_{\ph}^\Lambda
\end{align}

as the vacuum (vector). The photon Fock space $\fH_{\ph}^\Lambda$
comes equipped with the usual boson creation- and annihilation
operators $a(f)$ and $a^*(f)$, for $f \in \fh_\ph^\Lambda$, which
fulfill the canonical commutation relations (CCR),

\begin{align} \label{eq-015}
\brak{a(f), a^*(g)} 
\ = \ 
\sca{f}{g} \mathbf{1}_{\fH_{\ph}} 
\, , \quad
\brak{a(f), a(g)} \ = \ \brak{a^*(f), a^*(g)} \ = \ 0 
\, , \quad
\forall \, f,g \in \fh_\ph^\Lambda \, ,
\end{align}

and furthermore constitute a Fock representation of these, since the
annihilation operators $a(f)$ act destructively on the vacuum,

\begin{align} \label{eq-016}
a(f) \Om \ = \ 0 \, , \quad
\forall \, f \in \fh_\ph^\Lambda \, .
\end{align}

For any $f \in \fh_\ph^\Lambda$ the operators $a^*(f)$ and
$a(f)$ are closable operators when defined on the domain

\begin{align} \label{eq-017}
D\big( a^*(f) \big) 
\ := \ 
D\big( a(f) \big) 
\ := \
D\of{\cN^{1/2}} \, ,
\end{align} 

where $\cN$ denotes the number operator on
$\fH_{\ph}^\Lambda$,

\begin{align} \label{eq-018}
\cN \ = \ 
\sum_{n = 0}^\infty n \bfone^{n \otimes}_{\fh_\ph^\Lambda} \, , \quad
D\of{\cN} 
\ = \
\set{\of{\psi_n}_{n = 0}^\infty \in \fH_{\ph}^{\Lambda}
}{\sum_{n = 0}^\infty n^2 \norm{\psi_n}^2 < \infty} \, .
\end{align}

In the following we always identify the operators $a(f)$ and $a^*(f)$
with their respective closures on this domain.

The free evolution of the radiation field is generated by the second
quantization $d\Gamma\brak{\om}$ of the dispersion relation
$\om\of{k}$ (see \cite[Section~10.7]{ReedSimonII1980}). We call it the
free photon energy (or Hamiltonian) and denote it by
$H_\ph$. It leaves the subspaces 
$\fF^{\of{n}}\big[ \fh_\ph^\Lambda \big]$ of $\fF[\fh_\ph^\Lambda]$
invariant. On $\fF^{\of{0}}\big[ \fh_\ph^\Lambda \big]$, it vanishes
identically; on $\fF^{\of{1}}\big[ \fh_\ph^\Lambda \big]$, it is the
multiplication operator by $\om$, and on 
$\fF^{\of{n}}\big[ \fh_\ph^\Lambda \big]$, for $n \geq 1$, it is given
by

\begin{align} \label{eq-019}
\sum_{k = 1}^n 
\bfone_{h^\Lambda_{\ph}}^{\of{n-k}\otimes} \otimes \om \otimes 
\bfone_{h^\Lambda_{\ph}}^{(k-1)\otimes} \, . 
\end{align}

$H_{\ph}$ is defined on $\fH^\Lambda_{\ph}$ and it is
self-adjoint on $D\of{H_{\ph}} = D\of{\cN}$, where we
recall that $\cN$ denotes the number operator on
$\fH^\Lambda_{\ph}$. Similarly for the one particle photon
momentum operator, which is the multiplication by the momentum
variable $k$, we define the field momentum operator on the Fock space
by

\begin{align} \label{eq-020}
P_{\ph} \ := \ \mathrm{d}\Gamma\brak{k}. 
\end{align}

The standard model of non-relativistic quantum electrodynamics
\cite{BachFroehlichSigal1995,BachFroehlichSigal1998a} is called the
Pauli--Fierz Hamiltonian. It is given by

\begin{align} \label{eq:Intro_DefinitionPFH}
H^{g,\Lambda}
\ := \
\frac{1}{2}\of{i \nabla_x + g \vec{\mathbb{A}}\of{x}}^2 + H_{\ph}
\end{align}

acting on the Hilbert space 

\begin{align} \label{eq-021}
\fH^{\Lambda}
\ := \
\fH_{\prtcl} \otimes \fH_{\ph}^\Lambda 
\ \cong \ L^2\of{\RR^3_x; \fH_{\ph}^\Lambda} \, .
\end{align}

Here the parameter $g \geq 0$ is called the coupling constant. Note
that we assume the particle mass to be given by $m = 1$, which amounts
to choosing a suitable system of units. The interaction of the
particle with the radiation field is governed by the quantized vector
potential $\mathbb{A}$ acting on a vector 
$\Psi \in L^2\of{\RR^3_x;\fh_\ph}$ pointwise as

\begin{align} \label{eq-022}
\of{\bbA\Psi}\of{x} \ = \ \bbA\of{x} \, \Psi\of{x}
\end{align}

with the operator

\begin{align} \label{eq-023}
\bbA\of{x} 
\ = \ 
a^*(G_x) + a(G_x) \, ,
\end{align}

where $G_x(\uk) = e^{-ikx} G(\uk)$ and $G : \uRR^3 \to \RR^3$
represents the coupling function given by

\begin{align} \label{eq-024}
G\of{\uk} 
\ = \
\abs{\om\of{k}}^{-1/2} \epsilon\of{k,\pm}  \, ,
\end{align}

with polarization vectors $\epsilon\of{k,\pm}$ in Coulomb gauge. That
is, $\epsilon\of{ \cdot ,\pm}: \RR^3 \to \bbS^2$ is a measurable map such
that, for almost all $k \in \RR^3$, the three vectors

\begin{align}\label{eq:Intro_PolarizationVectors}
\sset{k, \epsilon\of{k,+}, \epsilon\of{k,-}}
\end{align} 

form a right-handed orthonormal frame in $\RR^3$. Note that, for every
$x \in \RR^3$, $\bbA\of{x}$ is well-defined because 
$\uk \mapsto e^{ikx} G\of{\uk} \in \fh_\ph^\Lambda$, due to the
ultraviolet cut-off.
It is well known \cite{HaslerHerbst2008b,Hiroshima2000b} that the
Pauli--Fierz Hamiltonian \eqref{eq:Intro_DefinitionPFH} is 
self-adjoint on the domain $D\of{H_{\prtcl}} \otimes D\of{H_{\ph}}$.

The specific choice of the polarization vectors in Coulomb gauge has
no physical significance and is merely a parametrization of the vector
space of divergence-free square-integrable vector fields 
$\RR^3 \to \RR^3$ as an $L^2$-space of complex-valued functions. It is
convenient to assume that

\begin{align}\label{eq:Intro_PolarizationVectorsProperties}
{\epsilon}\of{ a k, \pm} \ = \ 
{\epsilon}\of{k, \pm}  \, , \quad
{\epsilon}\of{ - k, \pm} \ = \ 
-{\epsilon}\of{k, \pm}.  
\end{align} 

for every $a> 0$. Such polarization vectors indeed exist.  To see
this, we first make a measurable, but apart from that arbitrary,
choice of polarization vectors $\epsilon\of{k,\pm}$ on the upper
half-sphere

\begin{align} \label{eq-025}
\set{k \in \RR^3}{\abs{k} = 1, k_3 > 0} \, ,
\end{align}

satisfying the Coulomb gauge condition. For all other $k$ in
the upper half-space $\set{k \in \RR^3}{ k_3 > 0}$ we then define

\begin{align*}
\epsilon\of{k, \pm} \ := \ \epsilon\of{k /|k|,\pm} \, ,
\end{align*}

and for every $k$ in the lower half-space 
$\set{k \in \RR^3}{ k_3 < 0}$ we define

\begin{align} \label{eq-026}
\epsilon\of{k, \pm} \ := \ -\epsilon\of{-k,\pm} \, .
\end{align}

This defines a choice $\epsilon: \uRR^3 \to \bbS^2$ of polarization
vectors fulfilling \eqref{eq:Intro_PolarizationVectorsProperties} up
to the set $\set{\uk \in \uRR^3}{ k_3 = 0}$ of zero measure.

\paragraph*{Fiber decomposition.}
The Pauli--Fierz Hamiltonian \eqref{eq:Intro_DefinitionPFH} commutes
with the total momentum operator $P$, which is defined by

\begin{align} \label{eq-027}
P \ := \ -i\nabla + P_{\ph} \, .   
\end{align}

This is known as the translation-invariance of the Pauli--Fierz
Hamiltonian. Thus, we can simultaneously diagonalize both operators
and construct fiber Hamiltonians of the restriction of 
$H^{g, \Lambda}$ to the spectral subspaces of fixed spectral values 
$p \in \RR^3$ of the total momentum operator $P$. Realizing this
diagonalization by a unitary transformation (involving the Fourier
transform in the particle position variable), we obtain these fiber
Hamiltonians as

\begin{align} \label{eq-028}
H_p^{g,\Lambda} 
\ = \
\of{\vec{P}_f + g\bbA\of{0} - p}^2 + H_{\ph} \, . 
\end{align}

More precisely, there is a unitary transformation $U$ such that    

\begin{align} \label{eq-029}
U H^{g,\Lambda} U^* 
\ = \
\int_{\RR^3}^{\oplus} H^{g,\Lambda}_p \, d^3 p , 
\end{align}

where each operator $H^{g,\Lambda}_p$ acts on
$\fH^{\Lambda}_{\ph}$. In this paper we only need the explicit form of
the fiber Hamiltonians and, therefore, we omit the details of their
construction and rather refer the reader, e.g., to
\cite{BachBreteauxTzaneteas2013,
  BallesterosFaupinFroehlichSchubnel2015}.

\paragraph*{Bogoliubov--Hartree--Fock energy.} 
The ground state energy $E_{\mathrm{gs}}\of{H_{p}^{g,\Lambda}}$ of the
fiber Hamiltonian $H_{p}^{\Lambda,g}$ is obtained by minimizing the
expectation of $H_{p}^{\Lambda,g}$ in the finite energy density
matrices $\DM\of{\fH^\Lambda_p}$ over $\fH_p^\Lambda$,
i.e.,

\begin{align} \label{eq:Intro_DefinitionGroundStateEnergy}
E_\gs\of{H_{p}^{g,\Lambda}} 
\ := \ &
\inf \spec\of{H_{p}^{g,\Lambda}} 
\ = \ 
\inf_{\DM\of{\fH_p}} \Tr\brak{\rho H_{p}^{\Lambda,g}} \, , 
\\[1ex] \label{eq:Intro_DefinitionDensityMatrix}
\DM\of{\fH^\Lambda_p} 
\ := \ & 
\set{ \rho \in \cL^1\of{\fH^\Lambda_p}}{\rho \geq 0, 
\Tr\brak{\rho} \ = \ 1, \, \rho H_0^{0,\Lambda},  H_0^{0,\Lambda} \rho 
\in \cL^1\of{\fH^\Lambda_p} } \, ,
\end{align}

where $\cL^1\of{\fH} \subset \cB\of{\fH}$ denote the space trace-class
operators on a Hilbert space $\cH$. The ground state energy in
\eqref{eq:Intro_DefinitionGroundStateEnergy} cannot be explicitly
calculated. Estimating its value is a highly difficult problem because
in quantum electrodynamics an infinite number of particles is always
present and the Hamiltonian under consideration does not preserve the
particle number. Already in many-body quantum mechanics, where a
finite, but large, number of particles is considered, estimating the
ground state is a challenging problem. One customary method to
approximate the ground state is to restrict the set of density
matrices $\rho$ to a subset or superset, respectively, of simpler
structure in order to get an effective variational problem whose
solution approximates the original ground state energy and, in fact
yields upper or lower bounds, respectively. Specifically, the
restriction to quasifree states applies to both many-body quantum
mechanics and quantum field theory, and in
\cite{BachBreteauxTzaneteas2013, DerezinskiNapiorkowskiSolovej2013} it
was shown that without loss of generality, the variation may even be
restricted to pure quasifree state.  In this paper, we use this
formalism and briefly present the definition of the concepts for the
convenience of the reader. There are many sources in the literature
addressing quasifree states, e.g., \cite{BachLiebSolovej1994,
  Solovej2014, BratteliRobinson-II-1996}. Here we follow 
\cite[Sect.~2 and 3]{BachBreteauxTzaneteas2013}.

If the infimum in \eqref{eq:Intro_DefinitionGroundStateEnergy} is
taken only over the set of pure quasi-free density matrices (see
\cite{BachBreteauxTzaneteas2013}), one obtains an approximation of the
ground state energy called the Bogoliubov--Hartree--Fock (BHF) energy
$E_\BHF\of{H_{p}^{g,\Lambda}}$ of the fiber Hamiltonian,

\begin{align} \label{eq:Intro_DefinitionBHFEnergy}
E_\BHF\of{H_{p}^{g,\Lambda}} 
\ := \ &
\inf_{\pQDM\of{\fH^\Lambda_p}} \Tr\brak{\rho H_{p}^{g,\Lambda}} \, , 
\\[1ex]  \label{eq:Intro_Definition-pQDM}
\pQDM\of{\fH} 
\ := \ & 
\set{\rho \in \DM\of{\fH}}{\rho 
\textnormal{ is pure and quasifree}} \, .
\end{align}

Note that the BHF-energy is always an upper bound on the true ground
state energy, that is

\begin{align} \label{eq-030}
E_\gs\of{H_{p}^{g,\Lambda}} 
\ \leq \ 
E_\BHF\of{H_{p}^{g,\Lambda}} \, .
\end{align}

The set $\pQDM$ admits the following parametrization
\cite{BachBreteauxTzaneteas2013} in terms of homogeneous Bogolubov
transformations and Weyl transformations,

\begin{align} \label{eq:Intro_ParametrizationPQDM}
\pQDM
\ = \
\set{\ket{\bbW_\eta\UU_B\Om}\bra{\bbW_\eta\UU_B\Om} \, }{ 
\; B \in \Bog\of{h^\Lambda_\ph}, \eta \in h^\Lambda_\ph} \, .
\end{align}

For our purposes it suffices to recall that:
\begin{itemize}
\item The homogeneous Bogolubov transforms are unitary operators on
  the photon Fock space $\fH_{\ph}^\Lambda$ which are
  parametrized by the set of Bogolubov matrices

\begin{align} \label{eq-031}
\Bog\of{\fh_\ph^\Lambda} 
& \ := \
\set{B = 
\begin{pmatrix}
U & \sfJ V \sfJ \\
V & \sfJ U \sfJ
\end{pmatrix}
}{U \in \cB\of{\fh_\ph^\Lambda}, \
V \in \HS\of{\fh_\ph^\Lambda}, \ B^* S B = S} \, ,
\end{align}

with

\begin{align} \label{eq-032}
S \ := \ 
\begin{pmatrix}
\bfone_{\fH_{\ph}^\Lambda} & 0 \\
0 & -\bfone_{\fH_{\ph}^\Lambda}
\end{pmatrix}  \, ,
\end{align}

where the Bogolubov transformation $\UU_B$ corresponding to 
$B \in \Bog\of{\fh_\ph^\Lambda}$ acts on the field operators as

\begin{align} \label{eq:Intro_BogAction}
\UU_B^{*} \of{a^*(f) + a(g)} \UU_B
\ = \ 
a^*\of{U f + \sfJ V \sfJ g} + a\of{V f + \sfJ U \sfJ g}
\end{align}

for all $f,g \in \fh_\ph^\Lambda$. In general, $\sfJ$ is an
arbitrary anti-unitary involution on $\fh_\ph^\Lambda$ (the set
pQDM does not depend on $\sfJ$). Here we take the specific choice

\begin{align} \label{eq:Intro_JChoice}
\sfJ \psi\of{k, \pm} 
\ = \ 
\overline{\psi\of{-k, \pm}} 
\quad \forall \psi \in \fh_\ph^\Lambda \, ,
\end{align}

because we make use of results of \cite{Herdzik2023}, where this
assumption is made.

\item The Weyl transformations are unitary operators on the photon
  Fock space $\fH_{\ph}^\Lambda$ which are parametrized by the
  vectors $\eta \in \fh_\ph^\Lambda$. They act on the field
  operators as

\begin{align} \label{eq:Intro_WeylAction}
\bbW^{*}_{\eta} \of{a^*(f) + a(g)} \bbW_{\eta}
\ = \
a^*(f) + \sca{f}{\eta} + a(g) + \sca{g}{\sfJ \eta}  \, ,
\end{align}

for all $f,g \in \fh_\ph^\Lambda$.
\end{itemize}

For a detailed discussion of both the Bogoliubov and Weyl transforms
we refer the reader to \cite{Berezin1966, Araki1971,
  ArakiShiraishi1971, Derezinski2006}.

The parametrization given in \eqref{eq:Intro_ParametrizationPQDM}
implies that the BHF-energy of $H_{p}^{g,\Lambda}$ is given by

\begin{align} \label{eq:Intro_FirstBHFParametrization}
E_\BHF\of{H^{g,\Lambda}_{p}} 
\ = \
\inf\Big\{ \sca{\bbW_\eta \UU_B \Om}{H^{g,\Lambda}_{p} \bbW_\eta \UU_B \Om}
\; \Big| \ B \in \mathrm{Bog}\of{\fh_\ph^\Lambda}, \ 
\eta \in \fh_\ph^\Lambda \Big\} \, . 
\end{align}

It is this formulation of the BHF-energy which was used in
\cite{BachBreteauxTzaneteas2013} and \cite{Herdzik2023} to analyze the
BHF-energy of the zero total momentum fiber Hamiltonian

\begin{align} \label{eq-033}
H^{g,\Lambda}_{0} 
\ = \
\of{P_{\ph} + g \bbA\of{0}}^2 + H_{\ph} \, .
\end{align}

The next section provides a brief overview of this analysis.

\paragraph*{Derivation of the Energy Functional.}
By explicitly calculating the scalar product on the right side of
Eq.~\eqref{eq:Intro_FirstBHFParametrization} using the actions of the
Bogolubov and Weyl transformations as outlined in
Eq.~\eqref{eq:Intro_BogAction} and \eqref{eq:Intro_WeylAction}, as
well as the Wick theorem \cite{PeskinSchroeder2018}, the author of
\cite{Herdzik2023} obtains the following first result.

\begin{prop} \label{prop:Intro_OriginalProblem}
The BHF-energy of $H_{0}^{g,\Lambda}$, as determined in
\eqref{eq:Intro_FirstBHFParametrization}, is given by

\begin{align} \label{eq-034}
E_\BHF\of{H^{g,\Lambda}_{0}}
\ = \ &
\inf\Big\{ \tcE_{g,\Lambda}\of{U,V,\eta} \, \Big| \ 
B = \of{
\begin{smallmatrix} U & \sfJ V \sfJ \\ V & \sfJ U \sfJ \end{smallmatrix}} 
\in \Bog\of{\fh_\ph^\Lambda}, \ 
\eta \in \fh_\ph^\Lambda \Big\} \, ,
\end{align}

where the functional $\tcE_{g,\Lambda}$ is given by 

\begin{align} \label{eq:Intro_RawEnergyFunctional}
\tcE_{g,\Lambda}\of{U,V,\eta} 
\ := \ & 
\frac{1}{2}\sum_{\nu = 1}^3 \of{\Tr\brak{k_\nu V^*V} 
+ \sca{\eta}{k_\nu\eta} + 2 \rRe \sca{\eta}{G_\nu}}^2 
\\ &
+ \frac{1}{2}\sum_{\nu = 1}^3 \Tr\brak{\of{V^* \sfJ U k_\nu}^2}
+ \Tr\brak{k_\nu V^* V k_\nu \of{1 + V^*V}}
\\ &
+ \frac{1}{2}\sum_{\nu = 1}^3  
\rRe\sca{G_\nu + k_\nu \eta}{\of{1 + 2 V^*V + 2V^* \sfJ U}\of{G_\nu + k_\nu \eta}}
\\ &
+ \Tr\brak{\abs{k}V^*V} + \sca{\eta}{\abs{k}\eta} \, .
\end{align}

\end{prop}

In \cite{Herdzik2023}, a reparametrization yielding a new effective
energy functional $\cE_{g, \Lambda}$ is introduced, which is used to
bound $\tcE_{g,\Lambda}\of{U,V,\eta}$ from above and below. This is
the contents of the next proposition.
 
\begin{prop} \label{prop:Intro_FullEnergyBounds}
The BHF-energy of $H^{g,\Lambda}_0$ equals and can be bounded from
above, respectively, by

\begin{align} \label{eq:Intro_FullEnergyBounds}
\inf\Big\{ & \cE_{g,\Lambda}\of{z,\eta} \: \Big| \
z \in \HS_0\of{\fh_\ph^\Lambda}, \ 
\eta \in \fh_\ph^\Lambda \Big\}  
\ = \ 
\\[1ex] \nonumber
& E_\BHF\of{H_{0}^{g,\Lambda}} 
\ \leq \ 
\inf\Big\{ \cE_{g,\Lambda}\of{z,\eta} \: \Big| \ 
z = \sfJ z \sfJ \in \HS_0\of{\fh_\ph^\Lambda}, \ 
\eta = \sfJ \eta \in \fh_\ph^\Lambda \Big\} \, , 
\end{align}

where 
\begin{align}  \label{eq:Intro_DefinitionFullEnergyFunctional} 
\cE_{g,\Lambda}\of{z,\eta} 
\ := \
& \frac{1}{8}\sum_{\nu = 1}^3 
\Tr\brak{k_\nu^2 z^2 \of{1 + z}^{-1} - k_\nu z k_\nu z \of{1 + z}^{-1}}
\nonumber \\ & + 
\frac{1}{2}\sum_{\nu = 1}^3 
\sca{G_\nu + k_\nu \eta}{\of{1 + z}^{-1}\of{G_\nu + k_\nu \eta}}
\nonumber \\ & +
\frac{1}{4} \Tr\brak{\abs{k} z^2 \of{1 + z}^{-1}} 
+ \sca{\eta}{\abs{k}\eta} \, .
\end{align}

Furthermore, $\cE_{g,\Lambda}$ exhibits the symmetry

\begin{align} \label{eq-035}
\cE_{g,\Lambda}\of{z,\eta} 
\ = \ 
\cE_{g,\Lambda}\of{\sfJ z \sfJ, \sfJ \eta}  \, , \quad 
\forall z \in \HS_0\of{\fh_\ph^\Lambda} \, , \,  
\eta \in \fh_\ph^\Lambda \, .
\end{align}

Hence, if $\cE_{g,\Lambda}$ posses a unique minimizer both sides of
\eqref{eq:Intro_FullEnergyBounds} coincide.
\end{prop}

In the proof of Proposition~\ref{prop:Intro_FullEnergyBounds} in
\cite{Herdzik2023}, the specific choice of $\sfJ$ in
\eqref{eq:Intro_JChoice} is highly relevant. In particular, the fact
that the first term of $\tcE_{g,\Lambda}$ in
\eqref{eq:Intro_RawEnergyFunctional}, namely

\begin{align} \label{eq-036}
\frac{1}{2} \sum_{\nu = 1}^3 \of{\Tr\brak{k_\nu V^*V} + \sca{\eta}{k_\nu\eta} + 2 \rRe \sca{\eta}{G_\nu}}^2  \, ,
\end{align}

vanishes for $V$ and $\eta$ satisfying $\sfJ V \sfJ = V$ and 
$\sfJ \eta = \eta$. 

\subsection{Main Results} \label{subsec:MainResults}

The main result of \cite{Herdzik2023} is the derivation of the energy
functional $\cE_{g, \Lambda}$ and the bounds
\eqref{eq:Intro_FullEnergyBounds} it yields. The present article is
devoted to its analysis. The task of determining the global minimum of
$\cE_{g, \Lambda}$ is, however, still too complicated. One of the
reasons is that, for every $g \geq 0$, it is not globally convex. This
is the content of Theorem~\ref{thm:Intro_NonConvexity} below.

To deal with a tractable model, we introduce a new effective energy
functional $\cG_{g, \Lambda}$ (see Definition~\ref{def:Intro_G} below)
in this work that simplifies $\cE_{g, \Lambda}$. Although there is
some loss of information when replacing $\cE_{g, \Lambda}$ by 
$\cG_{g, \Lambda}$, the new energy functional has the crucial advantage
of being globally convex (see Theorem~\ref{thm:Intro_ConvexityOfG}).

In this paper, we analyze the effective energy functional 
$\cG_{g, \Lambda}$ in detail. In
Theorem~\ref{thm:Intro_ExistenceOfMinimizer}, we state the existence
of a unique minimizer with certain properties. In order to obtain
explicit formulas for the minimum, we eliminate the variable $z$ and
derive an equivalent minimization problem in which only the variable
$\eta$ is considered. This is the content of
Theorem~\ref{thm:Intro_ReducedFunctional}. Finally, an explicit
estimation of the minimum is given in
Theorem~\ref{thm:Intro_GroundStateEnergyEstimate}.
Theorems~\ref{thm:Intro_ConvexityOfG},
\ref{thm:Intro_ExistenceOfMinimizer},
\ref{thm:Intro_ReducedFunctional}, and
\ref{thm:Intro_GroundStateEnergyEstimate}, together with
Theorem~\ref{thm:Intro_NonConvexity}, constitute the main results of
the present article.
 
\begin{thm}\label{thm:Intro_NonConvexity}
For any fixed $\eta \in \fh_\ph^\Lambda$ and $g \geq 0$ the functional 
$\cE_{g,\Lambda}\of{z,\eta}$ is not convex in $z$ on 
$\HS_{0}\of{\fh_\ph^\Lambda}$. Moreover, the term
 
\begin{align} \label{eq:Intro_NonConvexTerm}
\frac{1}{8}\sum_{\nu = 1}^3 
\Tr\brak{k_\nu^2 z^2\of{1 + z}^{-1} - k_\nu z k_\nu z\of{1 + z}^{-1}}
\end{align}

is not convex on $\HS_{0}\of{\fh_\ph^\Lambda}$.
\end{thm}

It is well know that the convexity of a function is, in general, not
preserved under reparametrizations. The non-convexity of
$\cE_{g,\Lambda}$ might be a consequence of an unwise choice of
reparametrisation in
Proposition~\ref{prop:Intro_FullEnergyBounds}. The present form of the
functional $\cE_{g,\Lambda}$, as given in
\eqref{eq:Intro_DefinitionFullEnergyFunctional}, has, however, the
distinct advantage that the term

\begin{align} \label{eq-037}
\frac{1}{2}\sum_{\nu = 1}^3 
\sca{G_\nu + k_\nu \eta}{\of{1 + z}^{-1}\of{G_\nu + k_\nu \eta}} \, ,
\end{align}

which captures the entire dependency on the coupling functions and
coupling constant, is indeed convex, as we prove in
Theorem~\ref{thm:convexity}.

The term \eqref{eq:Intro_NonConvexTerm} is the source of non-convexity
of the energy functional $\cE_{g,\Lambda}$. Actually, if we subtract
it from the energy, we obtain a convex functional, see
Theorem~\ref{thm:convexity}. We denote this functional by
$\tcG_{g,\Lambda}$.

\begin{defn}\label{def:Intro_GTilde}
Let $ 0 < \epsilon < 1$. The energy functional $\tcG_{g,\Lambda}:
\HS_{\epsilon}\of{\fh_\ph^\Lambda} \oplus \fh_\ph^\Lambda \to \RR$ 
is defined by

\begin{align} \label{eq-038}
\tcG_{g,\Lambda}\of{z,\eta}
\ = \
& \frac{1}{2} \sum_{\nu = 1}^3 
\sca{G_\nu + k_\nu \eta}{\of{1 + z}^{-1}\of{G_\nu + k_\nu \eta}}
\nonumber \\ 
& + \frac{1}{4} \Tr\brak{\abs{k}z^2\of{1 + z}^{-1}}
+ \sca{\eta}{\abs{k}\eta} \, .
\end{align}

\end{defn}

Note that $\tcG_{g,\Lambda}$ indeed results from
$\cE_{g,\Lambda}$ by subtracting the non-convex term

\begin{align} \label{eq-039}
\frac{1}{8}\sum_{\nu = 1}^3 
\Tr\brak{k_\nu^2 z^2\of{1 + z}^{-1} - k_\nu z k_\nu z\of{1 + z}^{-1}} \, .
\end{align}

It is important to notice that, though simplified, $\tcG_{g,\Lambda}$
still contains the interaction term

\begin{align} \label{eq-040}
\frac{1}{2}\sum_{\nu = 1}^3 
\sca{G_\nu + k_\nu \eta}{\of{1 + z}^{-1}\of{G_\nu + k_\nu \eta}} \, .
\end{align}

As a final preparation, we rescale the photon momentum in such a way
that the photon energy $\Lambda $ corresponds to the photon energy $1$
in the new scale, i.e., in new energy units. The rescaled energy
functional, denoted by $\cG_{g,\Lambda}$, is derived in
Section~\ref{sec:Properties}. In the next definition we give an
explicit formula, for which we here and henceforth assume that 
$\sigma \in (0,1)$ is fixed.

\begin{defn}\label{def:Intro_G}
For $\Lambda \geq 1$, we define the Hilbert space $\fh_\Lambda$ by

\begin{align} \label{eq-041}
\fh_\Lambda \ := \ L^2\of{\uS\of{\sigma/\Lambda,1};\CC} \, ,
\end{align}

and we set 
$\cG_{g,\Lambda}: \HS_\epsilon\of{\fh_\Lambda} \oplus \fh_\Lambda \to \RR$ 
to be

\begin{align} \label{eq-042}
\cG_{g,\Lambda}\of{z,\eta}
\ = \ & 
\frac{\Lambda^2}{2}\sum_{\nu = 1}^3 
\sca{G_\nu + k_\nu \eta}{\of{1 + z}^{-1}\of{G_\nu + k_\nu \eta}}
\nonumber \\ & 
+ \frac{\Lambda}{4} \Tr\brak{\abs{k}z^2\of{1 + z}^{-1}}
+ \Lambda\sca{\eta}{\abs{k}\eta} \, .
\end{align}

Furthermore, we denote the corresponding infimum by

\begin{align} \label{eq-043}
E_\cG^{g,\Lambda} 
\ := \ 
\inf\Big\{ \cG_{g,\Lambda}\of{z,\eta} \; \Big| \ 
\of{z,\eta} \in \HS_{\epsilon}\of{\fh_\Lambda} \oplus \fh_\Lambda \Big\} \, .
\end{align}
\end{defn}

The energy functional $\cG_{g,\Lambda}$ and its infimum
$E_\cG^{g,\Lambda}$ are the main subjects of investigation in this
work. The next theorem, which is proven in
Theorem~\ref{thm:convexity}, establishes that the functional
$\cG_{g,\Lambda}$ is strictly convex.

\begin{thm} \label{thm:Intro_ConvexityOfG}
The functional $\cG_{g,\Lambda}$ is strictly convex on 
$\HS_{\epsilon}\of{\fh_\Lambda} \oplus \fh_\Lambda$ (recall 
\eqref{eq:Intro_DefinitionHilbertSchmidtSet}).
\end{thm}

From the next theorem, which is proven in Theorem~\ref{thm:existence},
we obtain the existence of a unique minimizer of $\cG_{g,\Lambda}$
together with some of its properties.

\begin{thm} \label{thm:Intro_ExistenceOfMinimizer}
The functional $\cG_{g,\Lambda}$ possesses a unique global minimzer
$\of{z_{g,\Lambda},\eta_{g,\Lambda}}$ on $\HS_\epsilon\of{\fh_\Lambda}
\oplus \fh_\Lambda$.  This minimizer belongs to $\HS_0\of{\fh_\Lambda}
\oplus \fh_\Lambda$, and it is $\sfJ$-invariant in the sense that
$\of{\sfJ z_{g,\Lambda} \sfJ, \sfJ \eta_{g,\Lambda}} =
\of{z_{g,\Lambda},\eta_{g,\Lambda}}$.
\end{thm}

In addition to this abstract result we explicitly calculate the
partial minimizer $z_*\of{\eta}$ of $\cG_{g,\Lambda}\of{z,\eta}$ as a
functional of $z$ for a fixed $\eta$; the details of the construction
of $z_*\of{\eta}$ are presented in
Section~\ref{sec:PartialMinimization}. With this partial minimizer at
hand we are able to perform a variable reduction in $\cG_{g,\Lambda}$
as follows:

\begin{thm} \label{thm:Intro_ReducedFunctional}
Minimizing $\cG_{g,\Lambda}\of{z,\eta}$ over
$\HS_{\epsilon}\of{\fh_\Lambda} \oplus \fh_\Lambda$ is equivalent to
minimizing $\cG_{g,\Lambda}\of{z_*\of{\eta},\eta}$ over
$\eta \in \fh_\Lambda$ in the sense that

\begin{align} \label{eq:Intro_EquivalentMinimization}
E_\cG^{g,\Lambda} 
\ = \
\inf\Big\{ \cG_{g,\Lambda}\of{z_*\of{\eta},\eta} \; \Big| \ 
\eta \in \fh_\Lambda \Big\} \, .
\end{align}

Furthermore, $\cG_{g,\Lambda}\of{z_*\of{\eta},\eta}$ has
the explicit form

\begin{align} \label{eq:Intro_PartialMinimizationForm}
\cG_{g,\Lambda}\of{z_*\of{\eta},\eta}
\ = \
\frac{\Lambda}{2}
\Tr\brak{\of{\Lambda \abs{k}^{1/2}P_\eta\abs{k}^{1/2} 
         + \abs{k}^2}^{1/2} - \abs{k}}
+ \Lambda \sca{\eta}{\abs{k}\eta} \, ,
\end{align}

where $P_\eta$ is the finite-rank operator given by

\begin{align}
P_\eta \ = \ 
2 \sum_{\nu = 1}^3 \ket{G_\nu + k_\nu \eta} \bra{G_\nu + k_\nu \eta} \, , \quad
\forall \eta \in \fh_\Lambda \, .
\end{align}
\end{thm}

Theorem~\ref{thm:Intro_ReducedFunctional} is proven in
Lemma~\ref{lem:partialMinimisation} and
Theorem~\ref{thm:reducedFunctional}. The reduced form of
$\cG_{g,\Lambda}$ provided in
Eq.~\eqref{eq:Intro_PartialMinimizationForm}, together with
Eq.~\eqref{eq:Intro_EquivalentMinimization}, allows to determine the
asymptotics of the minimal energy of $\cG_{g,\Lambda}$ to leading
order in $\Lambda \gg 1$. This is the content of the next theorem
which is the main result of this paper (see
Theorem~\ref{thm:groundStateEnergyEstimate} for the proof).

\begin{thm} \label{thm:Intro_GroundStateEnergyEstimate}
The leading order asymptotics of the minimal energy $E_\cG^{g,\Lambda}$ 
can be estimated by 

\begin{align} \label{eq:Intro_AsymptoticOfFunctional}
\frac{4\sqrt{\pi}}{\sqrt{3}} g \Lambda^{3/2}
\ \leq \
E_\cG^{g,\Lambda}
\ \leq \ 
4 \sqrt{3\pi} g \Lambda^{3/2} 
\end{align}

for all $g > 0$ and $\Lambda > \frac{3}{8 \pi} g^{-2}$.
\end{thm}

The condition $\Lambda > \frac{3}{8 \pi} g^{-2}$ is of technical
origin. Since our purpose is to understand the ultraviolet behavior,
$\Lambda $ tending to infinity, for fixed $g$, this constriction is
not relevant in this paper. The upper bound is established by directly
evaluating $\cG_{g,\Lambda}$ at the trial state
$\of{z_*\of{0},0}$. The proof of the lower bound is more involved and
makes use of the the following general trace estimates:

\begin{thm} \label{thm:Intro_RootUpperBound}
For any positive trace-class operator $A$ and every two bounded and positive operators $B^2 \leq C^2$, we have that

\begin{align} \label{eq-044}
\Tr\brak{\of{A^2 + C^2}^{1/2} - C}
\ \leq \ 
\Tr\brak{\of{A^2 + B^2}^{1/2} - B}
\ \leq \ 
\Tr\brak{A} \, .
\end{align}

In particular, the positive operator $\of{A^2 + B^2}^{1/2} - B$ is
trace class.
\end{thm}

Theorem~\ref{thm:Intro_RootUpperBound} is proven in
Theorem~\ref{thm:rootLowerBound} and \ref{thm:rootUpperBound} below.

\section{Properties of the Energy Functional 
$\boldsymbol{\cG_{g,\Lambda}}$: \\ Continuity, Coercivity and Convexity}
\label{sec:Properties}

Throughout the rest of this work, we fix a parameter

\begin{align} \label{eq-045}
\epsilon \in \of{0, \tfrac{1}{2}} \, .
\end{align} 

We consider the energy functional 
$\tcG_{g,\Lambda} : \HS_\epsilon\of{\fh_\ph^\Lambda} \oplus \fh_\ph^\Lambda \to \RR$ 
as introduced in Definition~\ref{def:Intro_GTilde} and perform a
rescaling procedure resulting in the energy functional
$\cG_{g,\Lambda} : \HS_\epsilon\of{\fh_\Lambda} \oplus \fh_\Lambda \to \RR$ 
as given in Definition~\ref{def:Intro_G}. We recall that the Hilbert
space $\fh_\Lambda$ is given by

\begin{align} \label{eq-046}
\fh_\Lambda 
\ = \
 L^2\of{\uS\of{\sigma / \Lambda, 1} ;\CC} 
\end{align}

and introduce the dilation operators $\vphi_\Lambda$ and
$\vphi^{*}_{\Lambda}$ given by

\begin{align} \label{eq-047}
& \vphi_{\Lambda}  : ~
\fh_\ph^\Lambda \to \fh_\Lambda \, , \quad
f\of{\uk} \mapsto \Lambda^{3 / 2}f\of{\Lambda \uk} \, ,
\\ \label{eq-048}
& \vphi^{*}_{\Lambda}  : ~
\fh_\Lambda \to \fh_\ph^\Lambda \, , \quad
g\of{\uk} \mapsto \Lambda^{- 3 / 2}g\of{\uk / \Lambda} \, .
\end{align}

As linear operators, $\vphi_{\Lambda}$ and $\vphi^{*}_{\Lambda}$ are
unitary and adjoint to one another. Indeed, for all 
$f \in \fh_\ph^\Lambda$ and $g \in \fh_\Lambda$ one easily verifies that

\begin{align} \label{eq:UnitaryPhi}
\sca{f}{\vphi^{*}_{\Lambda}g}_{\fh_\ph^\Lambda}
\ = \
\sca{\vphi_{\Lambda}f}{g}_{\fh_\Lambda} \, , \quad
\vphi^{*}_{\Lambda}\vphi_{\Lambda} \ = \ \bfone_{\fh_\ph^\Lambda} \, , \quad
\vphi_{\Lambda}\vphi^{*}_{\Lambda} \ = \ \bfone_{\fh_\Lambda} \, . 
\end{align}

From $\vphi_{\Lambda}$ one also obtains a unitary map between the
Hilbert spaces $\HS\of{\fh_\ph^\Lambda}$ and $\HS\of{\fh_\Lambda}$ via

\begin{align} \label{eq-049}
\Phi_\Lambda &:
\HS\of{\fh_\ph^\Lambda} \to \HS\of{\fh_\Lambda} \, , ~
z \mapsto \vphi_\Lambda z \vphi_\Lambda^{*} \, , \\
\Phi^{*}_\Lambda &: \HS\of{\fh_\Lambda} \to \HS\of{\fh_\ph^\Lambda} \, , ~
w \mapsto \vphi_\Lambda^{*} w \vphi_\Lambda \, .
\end{align}

Crucially, $\Phi_\Lambda$ and $\Phi_\Lambda^*$ conserve operator
inequalities, and we have that

\begin{align} \label{eq-050}
\Phi_\Lambda\of{\HS_{\epsilon}\of{\fh_\ph^\Lambda}}
\ = \ 
\HS_{\epsilon}\of{\fh_\Lambda} \, ,
\quad \quad
\Phi^*_\Lambda\of{\HS_{\epsilon}\of{\fh_\Lambda}}
\ = \ 
\HS_{\epsilon}\of{\fh_\ph^\Lambda} \, .
\end{align} 

The operator $\Phi_\Lambda$ acts on bounded multiplication operators
$m\of{k} \in \cB\of{\fh_\ph^\Lambda}$ as

\begin{align} \label{eq:ActionOfPhiOnM}
\Phi_\Lambda\of{m\of{k}} 
\ = \ 
\vphi_\Lambda m\of{k} \vphi_\Lambda^{*} 
\ = \
m\of{\Lambda k} \in \cB\of{\fh_\Lambda} \, .
\end{align}

Furthermore, the dilation operator $\vphi_\Lambda$ acts on the
coupling function $G\of{\uk}$ as

\begin{align} \label{eq:ActionOfPhiOnG}
\vphi_\Lambda G \of{\uk}
\ = \ 
\Lambda^{3/2}\abs{\om\of{\Lambda k}}^{-1/2}\epsilon\of{\Lambda k, \tau}
\ = \
\Lambda G\of{\uk} \, ,
\end{align}

due to the properties of the polarization vectors given in
\ref{eq:Intro_PolarizationVectorsProperties}.  

The functional $\cG_{g,\Lambda}$ results from $\tcG_{g,\Lambda}$ by
rescaling the photon momenta by $\vphi_\Lambda$ in the sense that, for all 
$\of{z,\eta} \in \HS_{\epsilon}\of{\fh_\ph^\Lambda} \oplus \fh_\ph^\Lambda$, 
we have

\begin{align} \label{eq-051}
\tcG\of{z,\eta} \ = \ {\cG}\of{\Phi_\Lambda\of{z},\vphi_\Lambda\of{\eta}}
\end{align}

This follows directly from the unitary nature of $\vphi_\Lambda$ and
$\vphi^*_\Lambda$ \eqref{eq:UnitaryPhi} together with the actions
outlined in \eqref{eq:ActionOfPhiOnM} and
\eqref{eq:ActionOfPhiOnG}. As the mappings $\Phi_\Lambda$ and
$\vphi_\Lambda$ are both bijetive, the variational problems
of minimizing $\tcG_{g,\Lambda}$ and $\cG_{g,\Lambda}$ are equivalent
and

\begin{align} \label{eq-052}
\inf_{\of{z,\eta} \in \HS_\epsilon\of{\fh_\ph^\Lambda} \oplus \fh_\ph^\Lambda} 
\tcG_{g,\Lambda}\of{z,\eta} 
\ = \
\inf_{\of{z,\eta} \in \HS_\epsilon\of{\fh_\ph^\Lambda} \oplus \fh_\ph^\Lambda} 
\cG_{g,\Lambda}\of{z,\eta} \, . 
\end{align}

The advantage of $\cG_{g,\Lambda}$ over $\tcG_{g,\Lambda}$ is that the
dependence on the ultraviolet cut-off $\Lambda$ is shifted from the
underlying Hilbert space into the functional itself, making the
dependency more explicit. Hence $\cG_{g,\Lambda}$ is the object of our
study henceforth.  We begin this study by establishing certain
properties of $\cG_{g,\Lambda}$ which are essential for the proof of
the existence of a unique minimizer in the coming section.

\subsection{Continuity}
\label{subsec:Continuity}

The first property of $\cG_{g,\Lambda}$ we establish is its 
continuity in $\of{z,\eta}$ on 
$\HS_{\epsilon}\of{\fh_\Lambda} \oplus \fh_\Lambda$.

\begin{lem} \label{lem-2.1}
The functional $\cG_{g,\Lambda}$ is continuous on
$\HS_{\epsilon}\of{\fh_\Lambda} \oplus \fh_\Lambda$.

\begin{proof}
Recall that $\cG_{g,\Lambda}$ is given by

\begin{align} \label{eq-053}
\cG_{g,\Lambda}\of{z,\eta}
\ = \ & 
\frac{\Lambda^2}{2} \sum_{\nu = 1}^3 
\sca{G_\nu + k_\nu \eta}{\of{1 + z}^{-1}\of{G_\nu + k_\nu \eta}}
\nonumber \\ &
+ \frac{\Lambda}{4} \Tr\brak{\abs{k}z^2\of{1 + z}^{-1}}
+ \Lambda\sca{\eta}{\abs{k}\eta} \, ,
\end{align}

for all 
$\of{z,\eta} \in \HS_{\epsilon}\of{\fh_\Lambda} \oplus \fh_\Lambda$. The
term $\Lambda \sca{\eta}{\abs{k}\eta}$ is clearly continuous, as
$\abs{k}$ is bounded as a multiplication operator. The continuity of
the term $\frac{\Lambda}{4} \Tr\brak{\abs{k}z^2\of{1 + z}^{-1}}$ is
obtained from the fact that it is Frechet differentiable according to
Lemma~\ref{lem:derivative}. It remains to show the continuity of
$\frac{\Lambda^2}{2}\sum_{\nu = 1}^3 \sca{G_\nu + k_\nu \eta}{\of{1 +
    z}^{-1}\of{G_\nu + k_\nu \eta}}$. 
Since the map $\eta \to G_\nu + k_\nu \eta$ is continuous, 
it suffices to prove that

\begin{align} \label{eq-054}
\HS_{\epsilon}\of{\fh_\Lambda} \oplus \fh_\Lambda \ni
(z,\kappa) \ \mapsto \ \sca{\kappa}{\of{1 + z}^{-1}\kappa} \in \RR
\end{align}

is continuous. To this end, suppose that 
$\kappa, \kappa' \in \fh_\Lambda$, $z, z' \in \HS_{\epsilon}\of{\fh_\Lambda}$. 
Using the second resolvent identity, we obtain

\begin{align} \label{eq-055}
\sca{\kappa'}{ (1 + z')^{-1} \kappa' } 
\, - \, & \sca{\kappa}{ (1 + z)^{-1} \kappa}
\ = \ 
\sca{\kappa'}{ (1 + z')^{-1} (z-z') (1 + z)^{-1} \kappa' } 
\nonumber \\ 
& \, + \, \sca{\kappa'-\kappa}{ (1 + z)^{-1} \kappa'}
\, + \, \sca{\kappa}{ (1 + z)^{-1} (\kappa'-\kappa)} 
\end{align}

and thus

\begin{align} \label{eq-056}
\Big| \sca{\kappa'}{ (1 + z')^{-1} \kappa' } 
& \, - \, \sca{\kappa}{ (1 + z)^{-1} \kappa} \Big|
\\ \nonumber 
& \ \leq \ 
4 \|\kappa'\|_{\fh_\Lambda}^2 \, \|z-z'\|_{\HS} 
+ 2 \big( \|\kappa'\|_{\fh_\Lambda} + \|\kappa\|_{\fh_\Lambda} \big) 
\|\kappa'-\kappa\|_{\fh_\Lambda} \, , 
\end{align}

using that 
$\|(1 + z)^{-1}\|_\op, \|(1 + z')^{-1}\|_\op \leq (1-\epsilon)^{-1} \leq 2$. 
This yields the desired continuity.
\end{proof}
\end{lem}

\subsection{Coercivity}
\label{subsec:Coercivity}

The second property of $\cG_{g,\Lambda}$ that we establish is the
straightforward, but nevertheless important, observation that
$\cG_{g,\Lambda}$ is coercive.

\begin{lem} \label{lem:coercive}
The functional $\cG_{g,\Lambda}$ is coercive in the sense that 

\begin{align} \label{eq-057}
\cG_{g,\Lambda}\of{z,\eta}
\ \geq \
\frac{1}{2} \min\sset{\norm{z}_{\HS}, \norm{z}_{\HS}^2} 
+ 4 \norm{\eta}^2_{\fh_\Lambda} \, ,
\end{align}

for all $\of{z,\eta} \in \HS_{\epsilon}\of{\fh_\Lambda} \oplus \fh_\Lambda$.

\begin{proof}
The asserted coercivity~\eqref{lem:coercive} is already guaranteed by
the two terms $4 \Lambda \sca{\eta}{\abs{k}\eta}$ and 
$\Lambda \Tr\brak{\abs{k}z^2\of{1 + z}^{-1}}$ because the other terms
in $\cG_{g,\Lambda}$ are positive.  Indeed, due to the infrared
condition

\begin{align} \label{eq-058}
\abs{k} \ \geq \ \frac{\sigma}{\Lambda} \, ,
\end{align}

we have

\begin{align} \label{eq-059}
4 \Lambda \sca{\eta}{\abs{k}\eta}
\ \geq \
4 \sigma \norm{\eta}_{\fh_\Lambda}^2 
\end{align}

for all $\eta \in \fh_\Lambda$, as well as 

\begin{align} \label{eq-060}
\Lambda \Tr\brak{\abs{k}z^2\of{1 + z}^{-1}} 
\ \geq \ 
\sigma \Tr\brak{z^2\of{1 + z}^{-1}} \, .
\end{align}

We conclude the statement with the estimate

\begin{align} \label{eq-061}
\Tr\brak{z^2\of{1 + z}^{-1}} 
\ \geq \ & 
\of{1 + \norm{z}_{\mathrm{op}}}^{-1} \norm{z}_{\HS}^2
\ \geq \ 
\of{1 + \norm{z}_{\HS}}^{-1} \norm{z}_{\HS}^2
\nonumber \\[1ex] 
\ \geq \ & 
\frac{1}{2}\min\sset{\norm{z}_{\HS}, \norm{z}_{\HS}^2 } \, ,
\end{align}

where we use
$\of{1 + z}^{-1} \geq \of{1 + \norm{z}_{\mathrm{op}}}^{-1}$ which 
follows from the spectral theorem. 
\end{proof}
\end{lem}

\subsection{Convexity}
\label{subsec:Convexity}

While continuity and coercivity also holds for the full energy
functional $\cE_{g,\Lambda}$, we demonstrate in
Theorem~\ref{thm:nonConvexity} that the same does not hold true for
(global) convexity. It is here that the advantage of the functional
$\cG_{g,\Lambda}$ becomes apparent.

\begin{thm} \label{thm:convexity}
The functional $\cG_{g,\Lambda}$ is strictly convex on 
$\HS_{\epsilon}\of{\fh_\Lambda} \oplus \fh_\Lambda$.

\begin{proof}
We examine the convexity of each term of $\cG_{g,\Lambda}$
separately. As the term

\begin{align} \label{eq-062}
\sca{ \eta}{\abs{k}  \eta   } 
\ \geq \ 
\frac{\sigma}{\Lambda} \sca{\eta}{\eta} 
\end{align}

is a strictly positive quadratic form, it is strictly convex. We
recall that, for a given interval $I \subset \RR$, a function 
$f: I \to \RR$ is called \textit{operator convex} on $I$ if, 
for every pair of bounded self-adjoint operators $A, B$ such 
that $\sigma(A), \sigma(B) \subset I$ and any $\alpha \in (0,1)$, 
it follows that

\begin{align} \label{eq-063}
f\big( \alpha A + (1- \alpha)B  \big) 
\ \leq \ 
\alpha f(A) + (1- \alpha) f(B) \, .  
\end{align}  

Next we  see that the term $\Lambda \Tr\brak{\abs{k} z^2 \of{1 + z}^{-1}}$
is convex on $\HS_{\epsilon}\of{\fh_\Lambda}$. Notice that 
the function 

\begin{align} \label{eq-064}
x^2\of{1 + x}^{-1} 
\ = \
\of{1 + x} + \of{1 + x}^{-1} - 2
\end{align}

is operator convex on the interval $\of{-\epsilon, \infty}$. Indeed,
this convexity is strict in $z$, due to the strict operator convexity
of the function $x \mapsto \of{1 + x}^{-1}$ in the sense that 
$Q_\alpha\of{z, z'} \geq 0$ and $Q_\alpha\of{z, z'} \neq 0$, 
for $z \neq z'$ and $\alpha \in \of{0,1}$, where

\begin{align} \label{eq-065}
Q_\alpha\of{z, z'} 
\ := \
\alpha [1 + z]^{-1} 
+ (1 - \alpha) [1 + z']^{-1} 
- [ 1 + \alpha z + (1-\alpha) z' ]^{-1} \ \geq \ 0 \, . 
\end{align}

For a proof of this fact see Theorem~\ref{thm:strictConvexity}. It
remains to show the convexity of the term

\begin{align} \label{eq-066}
2 \Lambda^2\sum_{\nu = 1}^{3}
\sca{{G}_{\nu} + k_\nu \eta}{\of{1 + z}^{-1}\of{G_{\nu} + k_\nu \eta}} 
\end{align}

in $z$ and $\eta$. For that it is enough to show the convexity of
$\of{z,\eta} \mapsto \sca{\eta}{(1+z)^{-1}\eta}$ on the set
$\HS_\epsilon\of{\fh_\Lambda} \oplus \fh_\Lambda$.  The following
proof is an adaptation of an argument presented in
\cite[Theorem~1]{LiebRuskai1974} in the context of operators under the
trace.  
\\[1ex] 
Let $\delta >0$. For any $\kappa_1, \kappa_2 \in \fh_\Lambda$ and 
$y_1 = y_1^* \geq \delta$, $y_2 = y_2^* \geq \delta$, we define the
two vectors

\begin{align} \label{eq-067}
\gamma_1 \ := \ y_1^{-1/2} \kappa_1 - y_1^{1/2} \gamma \, , \quad
\gamma_2 \ := \ y_2^{-1/2} \kappa_2 - y_2^{1/2} \gamma 
\ \in \ \fh_\Lambda \, ,
\end{align}

where $\gamma \in \fh_\Lambda$ is chosen later. A direct computation 
yields

\begin{align} \label{eq:expandedQuadratic}
0 \ \leq \ &
\sca{\gamma_1}{\gamma_1} + \sca{\gamma_2}{\gamma_2}
\\ \nonumber 
\ = \ &
\sca{\kappa_1}{y_1^{-1}\kappa_1}
+ \sca{\kappa_2}{y_2^{-1}\kappa_2}
- \sca{\kappa_1 + \kappa_2}{\gamma}
- \sca{\gamma}{\kappa_1 + \kappa_2}
+ \sca{\gamma}{\of{y_1 + y_2}\gamma}
\\ \nonumber 
\ = \ &
\sca{\kappa_1}{y_1^{-1}\kappa_1}
+ \sca{\kappa_2}{y_2^{-1}\kappa_2}
- \sca{\kappa_1 + \kappa_2}{\of{y_1 + y_2}^{-1}\of{\kappa_1 + \kappa_2}} 
+ R\of{\gamma} \, ,
\end{align}

where $R\of{\gamma}$ is given by 

\begin{align} \label{eq-068}
R\of{\gamma} 
\ = \ 
\norm{\of{y_1 + y_2}^{-1/2}\of{\kappa_1 + \kappa_2} 
- \of{y_1 + y_2}^{1/2} \gamma}^2 \, .
\end{align}

The remainder term $R\of{\gamma}$ vanishes when we choose $\gamma$ to be

\begin{align} \label{eq-069}
\gamma \ := \
\of{y_1 + y_2}^{-1}\of{\kappa_1 + \kappa_2} \, .
\end{align}

With this choice of $\gamma$ the inequality 
\eqref{eq:expandedQuadratic} reads

\begin{align} \label{eq:unscaledConvexity}
\sca{\kappa_1 + \kappa_2}{\of{y_1 + y_2}^{-1}\of{\kappa_1 + \kappa_2}}
\ \leq \
\sca{\kappa_1}{y_1^{-1}\kappa_1}
+ \sca{\kappa_2}{y_2^{-1}\kappa_2} \, .
\end{align}

If $\alpha \in (0,1)$ and $z_1, z_2 \in \HS_\epsilon\of{\fh_\Lambda}$
then $y_1 := \alpha (1+z_1) \geq \frac{\alpha}{2} \geq \delta$ and 
$y_2 := (1-\alpha) (1+z_1) \geq \frac{1-\alpha}{2} \geq \delta$,
for $\delta := \frac{1}{2} \min\{\alpha, 1-\alpha\} > 0$. 
Furthermore choosing $\kappa_1 := \alpha \eta_1$ and 
$\kappa_2 := (1-\alpha) \eta_2$, Eq.~\eqref{eq:unscaledConvexity} 
reads

\begin{align} \label{eq:scaledConvexity}
& \sca{\alpha \eta_1 + (1-\alpha) \eta_2}{
[1 + \alpha z_1 + (1-\alpha) z_2]^{-1} 
\of{\alpha \eta_1 + (1-\alpha) \eta_2}}
\\ \nonumber & \qquad \qquad \qquad \qquad 
\ \leq \
\alpha \sca{\eta_1}{(1+z_1)^{-1} \eta_1}
+ (1-\alpha) \sca{\eta_2}{(1+z_2)^{-1} \eta_2} \, ,
\end{align}

which is the asserted convexity of 
$\of{z,\eta} \mapsto \sca{\eta}{(1+z)^{-1}\eta}$.
\end{proof}
\end{thm}

The strict convexity of $\cG_{g,\Lambda}$, together with its
continuity, provides the necessary regularity to prove the existence
of a unique minimizer.

\section{Existence of the Ground State Energy}
\label{sec:Existence}

This short section is dedicated to proving the existence of a unique
global minimizer of $\cG_{g,\Lambda}$ trough abstract arguments. It is
noteworthy that the existence of this unique minimzer does not depend
on the choice of coupling constant $g \geq 0$ or ultraviolet cut-off
$\Lambda > 0$.

\begin{thm} \label{thm:existence}
The functional $\cG_{g,\Lambda}$ possesses a unique global minimzer
$\of{z_{g,\Lambda},\eta_{g,\Lambda}}$ on 
$\HS_\epsilon\of{\fh_\Lambda} \oplus \fh_\Lambda$. 
This minimizer is $\sfJ$-invariant, i.e.,  
$\of{\sfJ z_{g,\Lambda} \sfJ, \sfJ \eta_{g,\Lambda}} 
= \of{z_{g,\Lambda},\eta_{g,\Lambda}}$. Moreover,
$\of{z_{g,\Lambda},\eta_{g,\Lambda}} \in 
\HS_0\of{\fh_\Lambda} \oplus \fh_\Lambda$.

\begin{proof}
We begin the proof with a discussion of the domain of the functional. 
The set $\HS_{\epsilon}\of{\fh_\Lambda} \oplus \fh_\Lambda$ 
is clearly convex. Furthermore, 
$\HS_{\epsilon}\of{\fh_\Lambda}$ is a closed subset of 
$\HS\of{\fh_\Lambda}$, 
as it is an intersection of a family of closed sets. 
More precisely, 

\begin{align} \label{eq-070}
\HS_{\epsilon}\of{\fh_\Lambda} 
\ = \ 
\bigcap_{\psi \in \fh_\Lambda, \|\psi\|=1} 
s_\psi^{-1}\of{[-\epsilon,\infty)} \, ,
\end{align}

where the $s_\psi$'s are the continuous mappings defined as 

\begin{align} \label{eq-071}
s_\psi : ~ \HS\of{\fh_\Lambda} \to \RR 
~, \quad 
z \mapsto \sca{\psi}{z \psi} \, .
\end{align}

Hence, $\HS_{\epsilon}\of{\fh_\Lambda} \oplus \fh_\Lambda$
is closed. Since 
$\HS_{\epsilon}\of{\fh_\Lambda} \oplus \fh_\Lambda$ is
convex and closed, the geometric version of the Hahn-Banach theorem
implies that it is weakly closed.

From Lemma~\ref{thm:convexity}, we know that $\cG_{g,\Lambda}$ is
convex, and since it is continuous, as well, it is also weakly lower
semi-continuous (see \cite[Corollary~3.9]{Brezis2011}).  As
$\cG_{g,\Lambda}$ is also coercive, by Lemma~\ref{lem:coercive}, the
direct method of variational calculus guarantees the existence of a
minimizing pair $\of{z_{g,\Lambda},\eta_{g,\Lambda}} \in
\HS_{\epsilon}\of{\fh_\Lambda} \oplus \fh_\Lambda$. We prove this in
Theorem~\ref{thm:minimizer}. The uniqueness of this minimizer follows
from the strict convexity of $\cG_{g,\Lambda}$. The unique minimizing
operator $z_{g,\Lambda}$ must be positive because

\begin{align} \label{eq-072}
\cG_{g,\Lambda}\of{z,\eta} 
\ \geq \
\cG_{g,\Lambda}\of{\abs{z},\eta} 
\end{align}

for all $z \in \HS_{\epsilon}\of{\fh_\Lambda}$.
This follows from the spectral theorem which implies 
that $(1 + z)^{-1} \geq ( 1 + |z| )^{-1}$ and hence also  
$z^2 (1 + z)^{-1} \geq |z|^2 ( 1 + |z| )^{-1}$. Hence 
$\of{z_{g,\Lambda},\eta_{g,\Lambda}} \in 
\HS_{0}\of{\fh_\Lambda} \oplus \fh_\Lambda$.
Finally, the $\sfJ$-invariance of $\cG_{g,\Lambda}$, 
(here we use \eqref{eq:Intro_PolarizationVectorsProperties}): 

\begin{align} \label{eq-073}
\cG_{g,\Lambda}\of{z,\eta} 
\ = \
\cG_{g,\Lambda}\of{\sfJ z \sfJ, \sfJ \eta} \, ,
\end{align}

implies that 

\begin{align} \label{eq-074}
\of{z_{g,\Lambda},\eta_{g,\Lambda}} 
\ = \
\of{\sfJ z_{g,\Lambda} \sfJ, \sfJ \eta_{g,\Lambda}}  \, ,
\end{align}

due to the uniqueness of the minimizer.
\end{proof}
\end{thm}

After this rather abstract existence result we are interested in
providing a more quantitative analysis. Hence, the following two
sections are concerned with approximating the minimizer and deriving
explicit estimates for the minimal energy.

\section{Partial Minimization Problem \\ and Variable Reduction}
\label{sec:PartialMinimization}

Even though $\cG_{g,\Lambda}$ is considerably less complex then the
full energy functional $\cE_{g,\Lambda}$, calculating the exact unique
minimizer $\of{z_{g,\Lambda},\eta_{g,\Lambda}} $ of $\cG_{g,\Lambda}$
still appears to be rather complicated. For our purpose, however, it
is sufficient to calculate the exact partial minimizers of
$\cG_{g,\Lambda}$. By first partially minimizing $\cG_{g,\Lambda}$
over $z$ and only then over $\eta$, we are then able to simplify the
minimization problem sufficiently to determine the asymptotics of the
minimal energy to leading order in $\Lambda \gg 1$. We begin this
section by calculating the partial gradient of $\cG_{g,\Lambda}$ in
$z$.

\subsection{The Partial Gradient of the Energy Functional}
\label{subsec:Gradient}

For any Hilbert space $\fH$ and every Frechet-differentiable
function $F: \fH \to \CC$, we denote by $D_x F$ its Frechet
derivative at the point $x \in \fH$, i.e., the continuous linear
transformation satisfying

\begin{align} \label{eq-075}
F\of{x+y} \ = \ 
F\of{x} + D_x F y + o\of{\norm{y}}, 
\qquad \forall \, y \in \fH \, .   
\end{align}

We denote by $\nabla_x F$ the gradient of $F$ at $x$, i.e., the unique
vector $\nabla_x F \in \fH$ such that $\sca{\nabla_x F}{y} = D_x F y$,
for all $y \in \fH$. We use this notation for Hilbert--Schmidt
operators, $\fH = \HS\of{\fh_\Lambda}$, with inner product 
$\sca{A}{B} = \Tr\brak{A^* B}$.

\begin{lem} \label{lem:derivative}
For any fixed $\eta \in \fh_\Lambda$, the functional 
$\cG_{g,\Lambda}\of{z,\eta}$ is Frechet differentiable with respect to 
$z$ on $\HS_{0}\of{\fh_\Lambda}$ and its gradient 
is given by 

\begin{align} \label{eq-076}
\nabla_z \cG_{g,\Lambda}\of{z,\eta} 
\ = \ 
- \of{1 + z}^{-1} \Lambda \abs{k} \of{1 + z}^{-1} 
+ \Lambda \abs{k} 
- \of{1 + z}^{-1} \Lambda^2 P_\eta \of{1 + z}^{-1} \, ,
\end{align}

where $P_\eta$ denotes the finite-rank operator 

\begin{align} \label{eq-077}
P_\eta \ := \ 
2 \sum_{\nu = 1}^3 \ket{G_\nu + k_\nu \eta}\bra{G_\nu + k_\nu \eta} \, .
\end{align}

\begin{rem}
In the case of $\eta = 0$ the rank-three operator $P_0$ is given by 

\begin{align} \label{eq-078}
P_0 \ = \ 2 \sum_{\nu = 1}^3 \ket{G_\nu}\bra{G_\nu} \, ,
\end{align}

which is homogeneous of degree two in the coupling constant $g$.
\end{rem}

\begin{proof}
First we note that $\HS_0\of{\fh_\Lambda}$ is contained in the
interior of $\HS_\epsilon\of{\fh_\Lambda}$. More specifically, for
every $z \in \HS_0\of{\fh_\Lambda}$ and every 
$h \in \HS\of{\fh_\Lambda}$ with $\norm{h}_{\HS} < \epsilon$, it
follows that $z + h \in \HS_\epsilon\of{\fh_\Lambda}$. We calculate
the derivatives of each of the terms
contributing to $\cG_{g,\Lambda}$ separately, beginning with 
$\Lambda \, \Tr\brak{\abs{k}z^2\of{1 + z}^{-1}}$. For this we notice 
that $z^2 \of{1 + z}^{-1} = z - 1 + \of{1 + z}^{-1}$ and hence

\begin{align} \label{eq-079}
\of{z+h}^2 & \of{1 + z +h}^{-1} - z^2 \of{1 + z}^{-1} 
\ = \ 
h + \of{1 + z +h}^{-1} - \of{1 + z}^{-1} 
\nonumber \\ 
\ = \ &
h - \of{1 + z +h}^{-1} h \of{1 + z}^{-1} 
\nonumber \\ 
\ = \ &
\of{1 + z +h}^{-1} \brak{ \of{1+z+h} h \of{1+z} -h } \of{1 + z}^{-1} 
\nonumber \\ 
\ = \ &
\of{1 + z +h}^{-1} \of{zh + h^2 + hz + zhz + h^2z} \of{1 + z}^{-1} \, ,
\end{align}

by the second resolvent identity. 
Since $\| \of{1 + z +h}^{-1} \|_\op, \| \of{1 + z}^{-1} \|_\op \leq 2$
and the factor in the middle on the right side of \eqref{eq-079} is
trace class, uniformly as $h \to 0$, so is its left side. Taking traces
then yields

\begin{align} \label{eq-080}
& \Lambda \, \Tr\brak{ \abs{k} \of{z + h}^2 \of{1 + z + h}^{-1} }
- \Lambda \, \Tr\brak{ \abs{k} z^2 \of{1 + z}^{-1} }
\\ \nonumber 
& \qquad \ = \ 
\Lambda \, \Tr\brak{ \abs{k} \of{1 + z + h}^{-1} 
 \of{zh + h^2 + hz + zhz + h^2z} \of{1 + z}^{-1}  } 
\\ \nonumber 
& \qquad \ = \ 
\Lambda \, \Tr\brak{ \abs{k} \of{1 + z}^{-1} 
 \of{zh + hz + zhz} \of{1 + z}^{-1}  } 
\, + \, \cO\of{\|h\|_\HS^2} 
\\ \nonumber 
& \qquad \ = \ 
\Lambda \, \Tr\brak{ \of{1 + z}^{-1} 
 \of{z \abs{k} + \abs{k} z + z \abs{k} z} \of{1 + z}^{-1} \, h } 
\, + \, \cO\of{\|h\|_\HS^2} \, ,
\end{align}

using the cyclicity of the trace and the fact that $z$ and
$\of{1+z}^{-1}$ commute in the last step. This implies that 

\begin{align} \label{eq:firstDer} 
\nabla_z \of{\Lambda \, \Tr\brak{\abs{k} z^2 \of{1 + z}^{-1}} }
\ = \ &
\Lambda \, \of{1 + z}^{-1} 
\of{z \abs{k} + \abs{k} z + z \abs{k} z} \of{1 + z}^{-1} 
\nonumber \\[1ex]
\ = \ &
\Lambda \, \of{ \abs{k} 
- \of{1 + z}^{-1} \Lambda\abs{k} \of{1 + z}^{-1} } \, .
\end{align}

For the interaction term we notice that 

\begin{align} \label{eq-081}
2 \Lambda^2 \sum_{\nu = 1}^3 
\sca{G_\nu + k_\nu \eta}{\of{1 + z}^{-1}\of{G_\nu + k_\nu \eta}}
\ = \
\Lambda^2 \Tr\brak{P_\eta\of{1 + z}^{-1}} \, .
\end{align}

Again, two applications of the second resolvent identity quickly
reveal that for any small $h \in \HS\of{\fh_\Lambda}$ we have

\begin{align}
& \Lambda^2 \, \Tr\brak{P_\eta\of{1 + z + h}^{-1}}
- \Lambda^2 \, \Tr\brak{P_\eta\of{1 + z}^{-1}} 
\\ & \equationindent \ = \
-\Lambda^2 \, \Tr\brak{\of{1 + z}^{-1} P_\eta\of{1 + z}^{-1} h}  
+ \cO\of{\norm{h}_{\HS}^2}  \, ,
\end{align}

which immediately implies 

\begin{align} \label{eq:secondDer}
\nabla_z \of{2 \Lambda^2 \sum_{\nu = 1}^3 
\sca{G_\nu + k_\nu \eta}{\of{1 + z}^{-1}\of{G_\nu + k_\nu \eta}}}
\ = \
- \of{1 + z}^{-1} \Lambda^2 \, P_\eta \, \of{1 + z}^{-1} \, .
\end{align}

Of course the final term $\sca{\eta}{\abs{k}\eta}$ is independent of
$z$ and thus its derivative simply vanishes. Combining
\eqref{eq:firstDer} and \eqref{eq:secondDer} then yields the desired
result.
\end{proof}
\end{lem}

\subsection{Partial Minimization Problem}
\label{subsec:PartialMinimizationProblem}

Existence and uniqueness of the partial minimizers of
$\cG_{g,\Lambda}$ can be established by the same arguments as in
Theorem~\ref{thm:existence}. As $\cG_{g,\Lambda}$ is quadratic in
$\eta$, the partial minimizer with respect to $\eta$ can be explicitly
determined by completing a square. Moreover, the form of the partial
gradient of $G_{g,\Lambda}$ with respect to $z$ allows us to
explicitly determine the partial critical points in $z$.

\begin{lem} \label{lem:partialMinimisation}
For fixed $\eta \in \fh_\Lambda$, there is a unique minimizer
$z_*\of{\eta} \in \HS_{\epsilon}\of{\fh_\Lambda}$ of the 
functional $\cG_{g,\Lambda}\of{\cdot,\eta}$. Moreover, it is given by the
positive operator

\begin{align} \label{eq-082}
z_{*}\of{\eta}
\ = \
\abs{k}^{-1/2} \of{\abs{k}^{1/2} 
\of{\Lambda P_\eta + \abs{k}} \abs{k}^{1/2}}^{1/2} \abs{k}^{-1/2} - 1 \, .
\end{align} 

Furthermore, for any fixed 
$z \in \HS_{\epsilon}\of{\fh_\Lambda}$, there is a (unique)
minimizer $\eta \in \fh_\Lambda$ of the functional
$\cG_{g,\Lambda}\of{z,\cdot}$. It is given by

\begin{align} \label{eq-083}
\eta_{*}\of{z}
\ = \
- \Big( 2\Lambda \abs{k} + 
\Lambda^2 \sum_{\nu = 1}^3 k_\nu \of{1 + z}^{-1} k_\nu \Big)^{-1}
\Big( \Lambda^2 \sum_{\nu = 1}^3 k_\nu\of{1 + z}^{-1} G_\nu \Big) \, .
\end{align} 

In particular, the unique minimizing pair 
$\of{z_{g,\Lambda}, \eta_{g,\Lambda}}$ of $\cG_{g,\Lambda}$ 
solves the fixed point equation 

\begin{align} \label{eq-083,1}
\of{z_{g,\Lambda}, \eta_{g,\Lambda}} 
\ = \ 
\big( z_{*}\of{\eta_{g,\Lambda}},\eta_{*}\of{z_{g,\Lambda}} \big) \, .
\end{align}

\begin{rem} \label{rem-001}
Taking $A^2 = |k|^{1/2} \Lambda P_{ \eta } |k|^{1/2} $ and 
$B^2 = |k|^{2} $ in Theorem~\ref{thm:rootUpperBound}, shows that the
operator $z_{*}\of{\eta}$ is a positive trace-class operator.
\end{rem}

\begin{proof} 
We begin by fixing $\eta \in \fh_\Lambda$. 
Lemmata~\ref{thm:convexity} and \ref{lem:coercive} prove that
$\cG_{g,\Lambda}$ is strictly convex and coercive. Then, the same holds
true when one of the variables is fixed. Following the general
arguments in the proof of Theorem~\ref{thm:existence}, we show that
there is a unique minimizer of $\cG_{g,\Lambda}\of{\cdot, \eta}$ in
$\HS_0\of{\fh_\Lambda}$. In particular, this minimizer satisfies
the stationarity equation $\nabla_z\cG_{g,\Lambda}\of{z,\eta} = 0$,
which is given by

\begin{align} \label{eq:stationarity}
-\of{1 + z}^{-1} \abs{k} \of{1 + z}^{-1} + \abs{k} 
- \of{1 + z}^{-1} \Lambda P_\eta \of{1 + z}^{-1}
\ = \ 0 \, ,
\end{align}

according to Lemma~\ref{lem:derivative}. Multiplying both sides of
\eqref{eq:stationarity} by the invertible operator $1 + z$, we obtain
gives the equivalent equation

\begin{align} \label{eq:stationarity2}
\of{1 + z}\abs{k}\of{1 + z}
\ = \
\Lambda P_\eta + \abs{k} \, .
\end{align}

It is easy to verify that $z_{*}\of{\eta}$ indeed solves
Eq.~\eqref{eq:stationarity2}, meaning that it must be the unique
minimizer of $\cG_{g,\Lambda}\of{\cdot, \eta}$.

For fixed $z \in \HS_\epsilon\of{\fh_\Lambda}$, we have 

\begin{align} \label{eq:explicitEtaDependency}
\cG_{g,\Lambda}\of{z,\eta}
\ = \ &
4\Lambda^2 \sum_{\nu = 1}^{3} \rRe \sca{G_\nu}{\of{1 + z}^{-1}k_\nu\eta}
+ 2\Lambda^2 \sum_{\nu = 1}^{3} \sca{\eta}{k_\nu\of{1 + z}^{-1}k_\nu\eta} 
\\ &+
4 \Lambda \sca{\eta}{\abs{k} \eta} + \rho\of{z} \, , \label{eq:yeah}
\end{align}

where 

\begin{align} \label{eq-084}
\rho\of{z} \ = \ 
\Lambda \, \Tr\brak{ \abs{k} z^2 \of{1 + z}^{-1} } 
+ 2 \Lambda^2 \sum_{\nu = 1}^3 \sca{G_\nu}{\of{1 + z}^{-1} G_\nu}
\end{align}

does not depend on $\eta$. Introducing the positive invertible
operator $T \geq 4 \sigma$ and the function $\psi$ by

\begin{align} \label{eq-085}
T \ := \ 
4\Lambda \abs{k} + 2 \Lambda^2 \sum_{\nu = 1}^3k_\nu\of{1 + z}^{-1}k_\nu \, , 
\equationspacing
\psi \ := \ 
2 \Lambda^2 \sum_{\nu = 1}^3 k_\nu\of{1 + z}^{-1} G_\nu \, ,
\end{align}

allows us to write Eq.~\eqref{eq:explicitEtaDependency} 
as [here $T$ is built from the last term in 
Eq.~\eqref{eq:explicitEtaDependency} and \eqref{eq:yeah}]

\begin{align} \label{eq-086}
\cG_{g,\Lambda}\of{z,\eta} 
\ = \
\norm{T^{1/2} \eta}_{\fh_\Lambda}^2 
+ 2 \rRe \sca{\psi}{\eta} + \rho\of{z} \, .
\end{align}

This expression can now be minimized in $\eta \in \fh_\Lambda$ 
by simply completing the square. More precisely, 

\begin{align} \label{eq-087}
\cG_{g,\Lambda}\of{z,\eta} 
\ = \ &
\norm{T^{1/2}\eta}^2_{\fh_\Lambda} 
+ 2 \rRe \sca{\psi}{\eta} + \rho\of{z}
\nonumber \\ 
\ = \ &  
\norm{T^{1/2}\eta}^2_{\fh_\Lambda} + 2 \rRe \sca{\psi}{\eta}
+ \norm{T^{-1/2}\psi}^2_{\fh_\Lambda}
-\norm{T^{-1/2}\psi}^2_{\fh_\Lambda} + \rho\of{z}
\nonumber \\ 
\ = \ & 
\norm{T^{1/2}\eta + T^{-1/2}\psi}^2_{\fh_\Lambda} 
- \norm{T^{-1/2}\psi}^2_{\fh_\Lambda} + \rho\of{z}
\nonumber \\ 
\ \geq \ & 
- \norm{T^{-1/2}\psi}^2_{\fh_\Lambda} + \rho\of{z} \, ,
\end{align}

where equality holds if, and only if, 
$\eta = -T^{-1} \psi$. Hence the unique minimizing choice of 
$\eta$ is given by $\eta_{*}\of{z} = -T^{-1} \psi$, i.e.,

\begin{align} \label{eq-088}
\eta_{*}\of{z}
\ = \
- \Big( 2\Lambda \abs{k} 
+ \Lambda^2 \sum_{\nu = 1}^3k_\nu\of{1 + z}^{-1}k_\nu \Big)^{-1}
\Big( \Lambda^2 \sum_{\nu = 1}^3 k_\nu \of{1 + z}^{-1} G_\nu \Big) \, .
\end{align}

\end{proof}
\end{lem}

\subsection{Variable Reduction}
\label{subsec:VariableReduction}

We expect the minimization with respect to the variable $z$, which
originates from the homogenuous Boguliubov transformation in
Proposition~\ref{prop:Intro_OriginalProblem}, to determine the leading
order asymptotics of the minimal energy of $\cG_{g,\Lambda}$, while
the minimization with respect to $\eta$, which parametrizes the Weyl
operators in Proposition~\ref{prop:Intro_OriginalProblem}, is expected
to yield only lower-order corrections to these.  This expectation is
confirmed by our findings in Section~\ref{sec:GroundStateEnergy},
where we pass from the functional $G_{g,\Lambda}\of{z,\eta}$ of two
variables $z$ and $\eta$ to the partially minimized functional
$G_{g,\Lambda}\of{z_*\of{\eta},\eta}$ of only the Weyl variable $\eta$
and show that, up to a multiplicative constant, the asymptotics is
unchanged to leading order, if the minimizing $\eta$ is replaced by
$0$.

\begin{thm} \label{thm:reducedFunctional}
Minimizing $\cG_{g,\Lambda}\of{z,\eta}$ over
$\HS_{\epsilon}\of{\fh_\Lambda} \oplus \fh_\Lambda$ is
equivalent to minimizing $\cG_{g,\Lambda}\of{z_{*}\of{\eta},\eta}$
over $\eta \in \fh_\Lambda$ in the sense that

\begin{align} \label{eq:equivalentMinimization}
E_\cG^{g,\Lambda} 
\ = \
\inf_{\eta \in \fh_\Lambda} 
\big\{ \cG_{g,\Lambda}\of{z_{*}\of{\eta},\eta} \; \big| \ 
\eta \in \fh_\Lambda \big\} \, , 
\end{align}

where $z_{*}\of{\eta}$ is the positive trace-class operator
given by \eqref{eq-082}. Furthermore, the reduced functional 
$\eta \mapsto \cG_{g,\Lambda}\of{z_{*}\of{\eta},\eta}$ takes the
explicit form

\begin{align} \label{eq:partialMinimizationForm}
\cG_{g,\Lambda}\of{z_{*}\of{\eta},\eta}
\ = \
2 \Lambda \, \Tr\brak{
\of{\Lambda \abs{k}^{1/2}P_\eta\abs{k}^{1/2} + \abs{k}^2}^{1/2} - \abs{k}}
+ 4 \Lambda \sca{\eta}{\abs{k}\eta} \, .
\end{align}

\begin{proof}
Eq.~\eqref{eq:equivalentMinimization} follows from 

\begin{align} \label{eq-090}
\inf_{\of{z,\eta} \in \HS_{\epsilon}\of{\fh_\Lambda} \oplus \fh_\Lambda} 
\cG_{g,\Lambda}\of{z,\eta}
\ = \
\inf_{\eta \in \fh_\Lambda} 
\Big\{ \inf_{z \in \HS_{\epsilon}\of{\fh_\Lambda}} 
\cG_{g,\Lambda}\of{z,\eta} \Big\}
\ = \
\inf_{\eta \in \fh_\Lambda} \cG_{g,\Lambda}\of{z_{*}\of{\eta},\eta} \, .
\end{align}

We now determine the explicit form of
$\cG_{g,\Lambda}\of{z_{*}\of{\eta},\eta}$. To this end we 
abbreviate $z_{*}\of{\eta} =: z_*$ and use \eqref{eq-077} to
rewrite $\cG_{g,\Lambda}\of{z_*,\eta} - 4 \Lambda \sca{\eta}{\abs{k}\eta}$ 
as

\begin{align} \label{eq-091}
\cG_{g,\Lambda}\of{z_*,\eta} & - 4 \Lambda \sca{\eta}{\abs{k}\eta} 
\nonumber \\ 
\ = \ & 
\Lambda \, \Tr\brak{ \abs{k} z_*^2 \of{1 + z_*}^{-1} }
+ 2 \Lambda^2 \sum_{\nu = 1}^3 
\sca{G_\nu + k_\nu \eta}{\of{1 + z_*}^{-1}\of{G_\nu + k_\nu \eta}} 
\nonumber \\ 
\ = \ & 
\Lambda \, \Tr\brak{ \abs{k} z_*^2 \of{1 + z_*}^{-1} 
+ \Lambda \, P_\eta \of{ 1 + z_* }^{-1} } 
\nonumber \\ 
\ = \ & 
\Lambda \, 
\Tr\brak{ \of{ \abs{k} z_*^2 + \Lambda P_\eta } \of{1 + z_*}^{-1} } \, .
\end{align}

Inserting the stationarity condition \eqref{eq:stationarity2}
to substitute for $\Lambda P_\eta$ and additionally using
the cyclicity of the trace, we further obtain

\begin{align} \label{eq-092}
\Lambda \, &
\Tr\brak{ \of{ \abs{k} z_*^2 + \Lambda P_\eta } \of{1+z_*}^{-1} } 
\nonumber \\ 
\ = \ & 
\Lambda \, 
\Tr\brak{ \abs{k} 
\of{ z_*^2 \of{1+z_*}^{-1} + \of{1+z_*} - \of{1+z_*}^{-1} } } \, .
\end{align}

Now note that $\frac{x^2}{1+x} + 1 + x - \frac{1}{1+x} 
= \frac{x^2-1}{1+x} + 1 + x = 2x$ and hence

\begin{align} \label{eq-093}
\cG_{g,\Lambda}\of{z_*,\eta} 
\ = \ & 
2 \Lambda \, \Tr\brak{ \abs{k} \, z_* }
- 4 \Lambda \sca{\eta}{\abs{k}\eta} 
\\ \nonumber 
\ = \ & 
2 \Lambda \Tr\brak{ 
\of{\Lambda \abs{k}^{1/2} P_{\eta} \abs{k}^{1/2} 
+ \abs{k}^2}^{1/2} - \abs{k} } 
- 4 \Lambda \sca{\eta}{\abs{k}\eta} \, ,
\end{align}

using \eqref{eq-082} and observing that $z_*$ is actually 
trace class, according to Remark~\ref{rem-001}.
\end{proof}
\end{thm}

\section{Ground State Energy Estimates}
\label{sec:GroundStateEnergy}

Using the partially minimized form of $G_{g,\Lambda}$ derived in
Section~\ref{sec:PartialMinimization} we give a lower bound on the
asymptotics of the minimal energy $\inf \cG_{g,\Lambda}$. Afterwards
we prove a matching upper bound by considering a simple trial state.
Both bounds crucially rely on the trace inequalities proven in
Section~\ref{sec:KeyInequalities}.

\begin{lem} \label{lem:functionalLowerBound}
The minimal energy admits the lower bound

\begin{align} \label{eq-094}
E_\cG^{g,\Lambda}
\ \geq \
\frac{4\sqrt{\pi}}{\sqrt{3}} \, g \, \Lambda^{3/2} \, ,
\end{align}

for all $g > 0$ and $\Lambda > \frac{3}{8 \pi} g^{-2}$.

\begin{proof}
According to Theorem~\ref{thm:reducedFunctional} and
Theorem~\ref{thm:rootLowerBound} in combination with $\abs{k}^2 \leq 1$
it is

\begin{align} \label{eq-095}
E_\cG^{g,\Lambda}
\ = \ & 
\inf_{\eta \in \fh_\Lambda} \cG_{g,\Lambda}\of{z_{*}\of{\eta},\eta}
\ \geq \ 
2 \Lambda \inf_{\eta \in \fh_\Lambda}  
\Tr\brak{ \of{\Lambda \abs{k}^{1/2} P_{\eta}\abs{k}^{1/2} + \abs{k}^2 }^{1/2} 
- \abs{k} }
\nonumber \\ 
\ \geq \ & 
2 \Lambda \inf_{\eta \in \fh_\Lambda}  
\Tr\brak{\of{\Lambda \abs{k}^{1/2}P_{\eta}\abs{k}^{1/2} + 1}^{1/2} - 1} \, .
\end{align} 

Of course for any $\mu \in \{1, 2, 3\}$ we have

\begin{align} \label{eq-096}
\Lambda \abs{k}^{1/2}P_{\eta}\abs{k}^{1/2} 
& \ = \
2 \Lambda \sum_{\nu = 1}^3 
\ket{\abs{k}^{1/2}\of{G_\nu + k_\nu \eta}}
\bra{\abs{k}^{1/2}\of{G_\nu + k_\nu \eta}}
\\ & \ \geq \ 
2 \Lambda \,
\ket{\abs{k}^{1/2}\of{G_\mu + k_\mu \eta}}
\bra{\abs{k}^{1/2}\of{G_\mu + k_\mu \eta}} \, .
\end{align}

For the sake of brevity we write 

\begin{align} \label{eq-097}
P_\eta^\mu 
\ := \
\ket{\abs{k}^{1/2}\of{G_\mu + k_\mu \eta}}
\bra{\abs{k}^{1/2}\of{G_\mu + k_\mu \eta}} 
\end{align}

in the following. 

We are now going to prove that there is always one $\mu$ for which
$P_\eta^\mu$ has a large spectral gap, uniformly in $\eta$. First we
note that the spectrum of $P_\eta^\mu$ is given by

\begin{align} \label{eq-098}
\sigma(P_\eta^\mu)
\ = \
\Big\{ 0, \;
\big\| \abs{k}^{1/2} \of{G_\mu + k_\mu \eta} \big\|_{\fh_\Lambda}^2 
\Big\} \, .
\end{align}

Now we estimate 

\begin{align} \label{eq:uniformNormLowerBound}
\inf_{\eta \in \fh_\Lambda} \max_{\mu = 1,2,3} &
\norm{\abs{k}^{1/2}\of{G_\mu + k_\mu \eta}}^2_{\fh_\Lambda}
\ \geq \ 
\inf_{\eta \in \fh_\Lambda} \frac{1}{3} \sum_{\mu = 1}^3 
\norm{\abs{k}^{1/2}\of{G_\mu + k_\mu \eta}}^2_{\fh_\Lambda}
\nonumber \\ 
\ = \ & 
\inf_{\eta \in \fh_\Lambda} \frac{1}{3} 
\| |k|^{1/2} G + |k|^{1/2} \eta k  \|_{ \fh_\Lambda^3 }^2
\ = \ 
\inf_{\eta \in \fh_\Lambda} \frac{1}{3} 
\Big( \| |k|^{1/2} G \|_{ \fh_\Lambda^3 }^2 
+ \| |k|^{1/2} \eta k  \|_{ \fh_\Lambda^3 }^2 \Big)
\nonumber \\ 
\ = \ & 
\frac{1}{3} \| |k|^{1/2} G \|_{ \fh_\Lambda^3 }^2
\ = \ 
g^2\frac{8 \pi}{3} \, ,
\end{align}

where we use that 
$\sca{\abs{k}^{1/2} G\of{k}}{\abs{k}^{1/2} k} = 0$ due to the
properties of the Coulomb gauge, see
\eqref{eq:Intro_PolarizationVectors}. The above estimate ensures
that, for $\Lambda > \frac{3}{8 \pi} g^{-2} $ and any $\eta \in
\fh_\Lambda$ there is a coordinate direction 
$\mu\of{\eta} \in \{1, 2, 3\}$, such that
$\sigma(2\Lambda P_\eta^{\mu\of{\eta}}) \cap \of{0,2} = \emptyset$.

Now, we use that

\begin{align} \label{eq-099}
\sqrt{x + 1} - 1 \ \geq \ \frac{\sqrt{x}}{2}
\end{align}

holds true on $\sset{0} \cup \of{2,\infty}$ and estimate

\begin{align} \label{eq-100}
2 \Lambda \inf_{\eta \in \fh_\Lambda} & 
\Tr\brak{\of{\Lambda \abs{k}^{1/2}P_{\eta}\abs{k}^{1/2} + 1}^{1/2} - 1}
\ \geq \
2 \Lambda \inf_{\eta \in \fh_\Lambda}  
\Tr\brak{\of{2 \Lambda  P_\eta^{\mu\of{\eta}} + 1}^{1/2} - 1}
\nonumber \\ 
\ \geq \ &
\Lambda \inf_{\eta \in \fh_\Lambda} 
\Tr\brak{\of{2 \Lambda  P_\eta^{\mu\of{\eta}}}^{1/2}}
\ \geq \
\sqrt{2} \Lambda^{3/2} \inf_{\eta \in \fh_\Lambda} 
\norm{\abs{k}^{1/2}\of{G_{\mu\of{\eta}} + k_{\mu\of{\eta}} \eta}}_{\fh_\Lambda}
\nonumber \\ 
\ \geq \ &
\frac{4 \sqrt{\pi}}{\sqrt{3}} g \Lambda^{3/2} \, ,
\end{align}

where we made use of \eqref{eq:uniformNormLowerBound} in the last inequality.
\end{proof}
\end{lem}

It remains to show the matching upper bound on the minimal energy.
For that we consider the explicit trial state $\of{z_*\of{0},0}$,
yielding the main theorem of this work.

\begin{thm} \label{thm:groundStateEnergyEstimate}
The minimal energy fulfills the estimate 

\begin{align} \label{eq:asymptoticOfFunctional}
\frac{4\sqrt{\pi}}{\sqrt{3}} \, g \, \Lambda^{3/2}
\ \leq \
E_\cG^{g,\Lambda}
\ \leq \ 
4 \sqrt{3\pi} \, g \, \Lambda^{3/2} \, ,
\end{align}

for all $g > 0$ and all $\Lambda > \frac{3}{8 \pi} g^{-2}$.

\begin{proof}
The lower bound in \eqref{eq:asymptoticOfFunctional} has already been
established in Lemma~\ref{lem:functionalLowerBound}. For the upper
bound we consider the trial state $\of{z_{*}\of{0},0}$, see
\eqref{eq-082}. Using the representation \eqref{eq:partialMinimizationForm}
of $\cG_{g,\Lambda}\of{z_{*}\of{0},0}$ from
Theorem~\ref{thm:reducedFunctional} and
Theorem~\ref{thm:rootUpperBound} we obtain

\begin{align} \label{eq:upperBoundEstimate} 
\cG_{g,\Lambda}\of{z_{*}\of{0},0} 
\ = \ & 
2 \Lambda \, 
\Tr\brak{\of{\Lambda\abs{k}^{1/2}P_0\abs{k}^{1/2} + \abs{k}^2}^{1/2} 
- \abs{k}} 
\nonumber \\ 
\ \leq \ & 
2 \Lambda^{3/2} \, \Tr\brak{ \of{\abs{k}^{1/2}P_0\abs{k}^{1/2}}^{1/2} } 
 \\ & 
\ \leq \ 
2 \sqrt{3} \, \Lambda^{3/2} \of{\Tr\brak{\abs{k}^{1/2}P_0\abs{k}^{1/2}}}^{1/2}
\nonumber \\ 
\ \leq \ & 
2 \sqrt{3} \, \Lambda^{3/2} \of{\sum_{\nu = 1}^3 
\norm{\abs{k}^{1/2}G_{\nu}}^2_{\fh_\Lambda}}^{1/2}
\ = \
4 \sqrt{3\pi} \, g \, \Lambda^{3/2} \, .
\end{align}

In the above estimates we made use of the fact that $P_0$ is a
rank-three operator. That is, we take 
$S = \of{\abs{k}^{1/2}P_0\abs{k}^{1/2}}^{1/2}$, which is a rank-three
operator, and efine $Q$ to be the orthogonal projection onto the range
of $S$. By the Cauchy-Schwarz inequality we obtain

\begin{align} \label{eq-101} 
\Tr(S) \ = \ \Tr( P S ) 
\ = \ 
\la P | S \ra_{HS} 
\ \leq \ 
\| P \|_{HS} \, \| S \|_{HS} 
\ = \
\sqrt{3} \, \| S \|_{HS} \, .
\end{align}

\end{proof}
\end{thm}

\section{Key Operator Inequalities}
\label{sec:KeyInequalities}

In this section we prove the trace inequalities from
Theorem~\ref{thm:rootUpperBound} and \ref{thm:rootLowerBound} which
are used throughout Section~\ref{sec:GroundStateEnergy}. While these
results are in principle known (see
e.g. \cite{BirmanKoplienkoSolomyak1975,BachHach2022}) we provide
elementary proofs for the convenience of the reader. We begin with
some preliminary lemmas.

\begin{lem} \label{Alem:strongConvergence}
Let $\fH$ be a Hilbert space and 
$\of{B_n}_{n = 1}^\infty \in \cB(\fH)^\NN$ be a sequence of positive
and bounded operators, such that $B_n \too{s} B$, for some positive,
bounded operator $B \in \cB(\fH)$, where $\too{s}$ denotes strong
convergence. Then $B_n^2 \too{s} B^2$ and $B_n^{1/2} \too{s} B^{1/2}$.

\begin{proof}
Let $\psi \in \fH$ be any vector and denote $\phi := B \psi$. 
The uniform boundedness principle
implies that $( \|B_n\| )_{n \in \NN } \in (\RR_0^+)^\NN$ is
  bounded, and we denote $b := \sup_{n \in \NN} \|B_n\|$. Using that 

\begin{align} \label{eq-102} 
\| B_n^2 \psi - B^2 \psi \| 
\ = \ &
\| B_n (B_n - B) \psi + (B_n -B) \phi \|
\nonumber \\
\ \leq \ & 
b \, \| (B_n - B) \psi \| \, + \, \| (B_n -B) \phi \| \, , 
\end{align}

we conclude that $B_n^2 \psi \to B^2 \psi$ and hence 
$B_n^2 \too{s} B^2$.

To prove $B_n^{1/2} \too{s} B^{1/2}$, we first notice that, for
every $x \geq 0$, we have that $\frac{1}{\pi} \int_{0}^{\infty}
\frac{x}{x + \lambda} \lambda^{1/2} = \sqrt{x}$. This, together with
the spectral theorem (and Tonelli's Theorem), implies that the same
identity holds true weakly if we substitute $x$
by a positive operator $A$. It also holds true strongly because the
integral at stake converges strongly: the term $\Big \| \frac{A}{A +
  \lambda} \Big \| \lambda^{1/2}$ can be majorized by an integrable
function. This is a well-known identity for 
strictly positive operators, see, e.g., 
\cite[Section~1.4]{Caps2002}. With this,
the identity $\frac{x}{x+\lambda} = 1 - \frac{1}{x+\lambda}$,
and the second resolvent equation, we obtain the estimate

\begin{align} \label{eq-103}
\norm{\of{B^{1/2} - B_n^{1/2}} \psi}_\fH
\ \leq \ &
\int_{0}^\infty \norm{\of{\frac{B}{B + \lambda} 
- \frac{B_n}{B_n + \lambda}} \psi}_\fH \;
\frac{d\lambda}{\pi \, \sqrt{\lambda}} 
\nonumber \\ 
\ = \ &
\int_{0}^\infty \norm{\of{\frac{1}{B_n + \lambda} 
- \frac{1}{B + \lambda}} \psi}_\fH \;
\frac{d\lambda}{\pi \, \sqrt{\lambda}} 
\nonumber \\ 
\ = \ &
\int_{0}^\infty \norm{\of{\frac{1}{B_n + \lambda} 
\of{B - B_n} \frac{1}{B + \lambda}} \psi}_\fH \;
\frac{d\lambda}{\pi \, \sqrt{\lambda}} \, .
\end{align}

As a function of $\lambda$, the norm converges pointwise to zero due
to

\begin{align} \label{eq-104}
\norm{\of{\frac{1}{B_n + \lambda} \of{B - B_n} \frac{1}{B + \lambda}}\psi}_\fH 
\ \leq \ & 
\norm{\of{\frac{1}{B_n + \lambda}}}_{\mathrm{op}} \cdot
\norm{\of{B - B_n} \frac{1}{B + \lambda}\psi}_\fH 
\nonumber \\[1ex] 
\ \leq \ &
\frac{1}{\lambda} \;
\norm{\of{B - B_n} \frac{1}{B + \lambda}\psi}_\fH 
\too{n \to \infty} 0 \, ,
\end{align}

where we use $B_n \too{s} B$. Further note that

\begin{align} \label{eq-105}
\norm{ \frac{B}{B + \lambda} - \frac{B_n}{B_n + \lambda} }_\op
\ \leq \ 
\norm{\frac{B}{B + \lambda}}_\op 
+ \norm{\frac{B_n}{B_n + \lambda}}_\op 
\ \leq \ 2 \, ,
\end{align}

and 

\begin{align} \label{eq-106,1}
\norm{ \frac{B}{B + \lambda} - \frac{B_n}{B_n + \lambda} }_\op
\ \leq \ 
\norm{\frac{1}{B + \lambda}}_\op 
+ \norm{\frac{1}{B_n + \lambda}}_\op 
\ \leq \ \frac{2}{\lambda} \, ,
\end{align}

since $B_n, B \geq 0$, and $\lambda >0$. Therefore,

\begin{align} \label{eq-106,2}
X(\lambda) \ := \ 
\frac{2}{\pi} \: \Big( \frac{ \bfone[\lambda \leq 1] }{ \lambda^{1/2} } 
\, + \, \frac{ \bfone[\lambda > 1]}{\lambda^{3/2}} \Big)
\end{align}

is integrable and dominates the integrand in the last line of
\eqref{eq-103}. By the Lebesgue dominated convergence theorem, it
follows that

\begin{align} \label{eq-107}
\norm{\of{B^{1/2} - B_n^{1/2}} \psi}_\fH
\ = \
\frac{1}{\pi} \int_{0}^\infty 
\norm{\of{\frac{1}{B_n + \lambda} \of{B - B_n} 
\frac{1}{B + \lambda}}\psi}_\fH \;
\frac{d\lambda}{\pi \, \sqrt{\lambda}} 
\ \too{n \to \infty} \ 0 \, .
\end{align} 
\end{proof}
\end{lem}

\begin{lem} \label{Alem:traceApproximation}
Let $\fH$ be a Hilbert space and 
$\of{A_n}_{n = 1}^{\infty} \in \cL(\fH)^\NN$ be a sequence of positive
trace-class operators such that

\begin{align} \label{eq-108}
C \ := \ \sup_{n \in \NN} \: \Tr\brak{A_n} \ < \ \infty \, ,
\end{align} 

and suppose that $A_n \too{s} A$ on $\cD \subset \fH$, for some 
densely defined linear operator $A: \cD \to \fH$. 
Then $A$ is, or extends to, a positive trace-class operator with 

\begin{align} \label{eq-109}
\Tr\brak{A} \ \leq \ C \, .
\end{align}

\begin{proof}
For $\psi \in \cD$, we have that 
$\| A \psi \|_\fH = \lim_{n \to \infty} \| A_n \psi \|_\fH 
\leq C \| \psi \|_\fH$. Hence $A$ extends to a bounded operator 
which we denote by the same letter $A \in \cB(\fH)$. From
$\la \vphi | A \psi \ra = \lim_{n \to \infty} \la \vphi | A_n \psi \ra$
we also easily decduce that $A$ is self-adjoint and positive. 
We now show that $A$ is indeed trace class. 
For that let $\of{\vphi_\ell}_{\ell = 1}^\infty \subset \fH$ 
be an orthonormal basis. For any $L \geq 1$ it is 

\begin{align} \label{eq:sumConvergence}
\sum_{\ell = 1}^L \sca{\vphi_\ell}{A \vphi_\ell} 
\ = \ &
\lim_{n \to \infty} \bigg\{
\sum_{\ell = 1}^L \sca{\vphi_\ell}{A_n \vphi_\ell} \bigg\} 
\ \leq \ &
\sup_{n \to \infty} \bigg\{
\sum_{\ell = 1}^\infty \sca{\vphi_\ell}{A_n \vphi_\ell} \bigg\} 
\ = \ C \, ,
\end{align}

and we obtain the asserted trace-class property of $A$ in the limit 
$L \to \infty$.
\end{proof}
\end{lem}

\begin{lem} \label{Alem:traceEstimate}
Let $\fH$ be a Hilbert space, and assume that $A \in \cB(\fH)$ is a
positive bounded operator and $B$ is a densely defined, positive, and
invertible, operator such that $A \leq B$. Then it is

\begin{align} \label{eq-112}
A B^{-1} A \ \leq \ A \, .
\end{align}

In particular, if $A$ is trace class then so is $A B^{-1} A$.

\begin{proof}
We may write $B = A + X$ for some positive operator $X$. Instead of
$A$ we first consider the invertible operator $A + \epsilon$ for
$\epsilon > 0$. Due to the operator monotonicity of the inverse (see
\cite[Ch.~5]{Bhatia1996}) and the positivity of $A$, we have

\begin{align} \label{eq:epsilonEstimate}
\of{A + \epsilon} \of{A + \epsilon + X}^{-1} \of{A + \epsilon}
\ \leq \
\of{A + \epsilon} \of{A + \epsilon}^{-1} \of{A + \epsilon}
\ = \ 
A + \epsilon \, .
\end{align}

In the limit $\epsilon \to 0$ both sides of \eqref{eq:epsilonEstimate}
converge strongly to $A B^{-1} A$ and $A$, respectively. This yields
the desired inequality $A B^{-1} A \leq A$.
\end{proof}
\end{lem}

We now come to the main statements of the section.

\begin{thm} \label{thm:rootUpperBound}
For any positive trace class operator $A$ and every bounded positive
operator $B$ it is

\begin{align} \label{eq-113}
\Tr\brak{\of{A^2 + B^2}^{1/2} - B}
\ \leq \ 
\Tr\brak{A} \, .
\end{align}

In particular, the positive operator 
$\of{A^2 + B^2}^{1/2} - B$ is also trace class.

\begin{proof}
First we note that operators of the form $\of{A^2 + B^2}^{1/2} - B$
are always positive due to the operator monotonicity of the square
root (see \cite[Ch.~5]{Bhatia1996}). For now we assume that $A$ and
$B$ are positive operators on a finite-dimensional Hilbert space. Also
assume that $B \geq \epsilon > 0$. We write

\begin{align} \label{eq-114}
D_{\pm} \ := \
\of{A^2 + B^2}^{1/2} \pm B \, .
\end{align}

We may write the trace of $D_{-} = \of{A^2 + B^2}^{1/2} - B$ as

\begin{align} \label{eq-115}
\Tr\brak{\of{A^2 + B^2}^{1/2} - B}
\ = \
\Tr\brak{D_{-}D_{+}D_{+}^{-1}}
\ = \
\Tr\brak{D_{+}^{-1/2} A^2 D_{+}^{-1/2} + R} \, ,
\end{align}

where we expand the product

\begin{align} \label{eq-116}
D_{-}D_{+} 
\ = \
A^2 + \brak{\of{A^2 + B^2}^{1/2},B}
\end{align}

and define the operator $R$ by 

\begin{align} \label{eq-117}
R \ = \
D_{+}^{-1/2} \brak{\of{A^2 + B^2}^{1/2},B} D_{+}^{-1/2}
\ = \ 
D_{+}^{-1/2} \brak{D_{+},B} D_{+}^{-1/2} \, . 
\end{align}

As all operators live on a finite dimensional Hilbert space, it is
$\Tr\brak{R} = 0$ and thus

\begin{align} \label{eq-118}
\Tr\brak{\of{A^2 + B^2}^{1/2} - B} 
\ = \  
\Tr\brak{D_{+} ^{-1/2} A^2 D_{+} ^{-1/2} + R}
\ = \
\Tr\brak{D_{+} ^{-1/2} A^2 D_{+} ^{-1/2}} \, .
\end{align}

Note that $D_{+} > A$, again by the monotonicity of the
square-root. Hence Lemma~\ref{Alem:traceEstimate} yields

\begin{align} \label{eq-119}
\Tr\brak{D_{+}^{-1/2} A^2 D_{+}^{-1/2}}
\ = \
\Tr\brak{A D_{+}^{-1} A}
\ \leq \
\Tr\brak{A} \, .
\end{align}

We drop the assumption of $B \geq \epsilon > 0$ and only assume 
$B \geq 0$, but the Hilbert space remains finite-dimensional. In this
case we have that

\begin{align} \label{eq-120}
\Tr\brak{\of{A^2 + B^2}^{1/2} - B}
\ = \
\lim_{\epsilon \to 0} 
\Tr\brak{\of{A^2 + \of{B + \epsilon}^2 }^{1/2} - \of{B + \epsilon}}
\ \leq \ 
\Tr\brak{A} \, ,
\end{align}

due to $\of{A^2 + \of{B + \epsilon}^2 }^{1/2} - \of{B + \epsilon} 
\too{s} \of{A^2 + B^2 }^{1/2} - B$ 
(see Lemma~\ref{Alem:strongConvergence}) and 
Lemma~\ref{Alem:traceApproximation}.

Lastly we drop the remaining assumption and consider a general
positive trace-class operator $A$ and a bounded positive operator $B$.
Let $\{ \vphi_\ell \}_{\ell \in \NN} \subset \fH$ be an orthonormal
basis of the Hilbert space $\cH$, define 
$\fH_L := \huelle\{ \psi_1, \cdots, \psi_n \}$ to be the subspace
spanned by the first $L \in \NN$ vectors in that basis, and denote by
$Q_L$ the orthogonal projection onto $\fH_L$. We approximate $A$ and
$B$ by the positive finite-rank operators $A_L := Q_L A Q_L$ and $B_L
:= Q_L B Q_L$ in the strong sense. In this now finite-dimensional
approximation it is

\begin{align} \label{eq-121}
\Tr\brak{\of{A_L^2 + B_L^2}^{1/2} - B_L} 
\ \leq \ 
\Tr\brak{A_L}
\ \leq \ 
\Tr\brak{A} \, .
\end{align}

According to Lemma~\ref{Alem:strongConvergence} it also is 
$\of{A_n^2 + B_n^2}^{1/2} - B_n \too{s} \of{A^2 + B^2}^{1/2} - B$. 
This allows us to make use of Lemma~\ref{Alem:traceApproximation}
from which we conclude that $\of{A^2 + B^2}^{1/2} - B$ 
is trace class, indeed, and

\begin{align} \label{eq-122}
\Tr\brak{\of{A^2 + B^2}^{1/2} - B} \ \leq \ \Tr\brak{A} \, .
\end{align}

\end{proof}
\end{thm}

\begin{lem} \label{lem:operatorMon}
Let $A$ and $B$ be bounded and positive operators. Then for every $\epsilon > 0$ we have
\begin{align}
\of{A^2 + B^2 + \epsilon}^{1/2} - \of{B^2 + \epsilon} 
\ \leq \
\of{A^2 + B^2}^{1/2} - B \, .
\end{align}

\begin{proof}
It is sufficient to show that the difference operator 
\begin{align}
\of{B^2 + \epsilon}^{1/2} - B  + \of{A^2 + B^2}^{1/2} - \of{A^2 + B^2 + \epsilon}^{1/2}
\end{align}
is positive. 
For all real numbers $x \geq 0$ we have
\begin{align} \label{eq:scalarEq}
	x^{1/2} - \of{x + \epsilon}^{1/2} 
	\ = \ 
	-\frac{\epsilon}{\of{x + \epsilon}^{1/2} + x^{1/2}}
	\, .
\end{align}
By functional calculus we thus have
\begin{align}
	&\of{B^2 + \epsilon}^{1/2} - B + \of{A^2 + B^2}^{1/2} - \of{A^2 + B^2 + \epsilon}^{1/2}
	\\&\equationindent \nonumber \ = \
	\of{B^2 + \epsilon}^{1/2} - B - \epsilon \of{\of{A^2 + B^2 + \epsilon}^{1/2} + \of{A^2 + B^2}^{1/2}}^{-1} 
	\, .
\end{align}
The operator monotonicity of the square root and the inverse then yields
\begin{align}
	&\of{B^2 + \epsilon}^{1/2} - B - \epsilon \of{\of{A^2 + B^2 + \epsilon}^{1/2} + \of{A^2 + B^2}^{1/2}}^{-1} 
	\\&\equationindent \nonumber  \ \geq \
	\of{B^2 + \epsilon}^{1/2} - B - \epsilon \of{\of{B^2 + \epsilon}^{1/2} + B}^{-1} 
	\ = \
	0
	\, ,
\end{align}
where the last equality again follows from \eqref{eq:scalarEq}.
\end{proof}
\end{lem}

\begin{thm} \label{thm:rootLowerBound}
For any trace-class operator $A$ and any two bounded and strictly positive
operators $B$ and $C$ such that $B^2 \leq C^2$ we have

\begin{align} \label{eq:rootLowerBound}
\Tr\brak{\of{A^2 + C^2}^{1/2} - C}
\ \leq \ 
\Tr\brak{\of{A^2 + B^2}^{1/2} - B} \, .
\end{align}

\begin{proof}
Assume first that $B^2 \geq \epsilon$ for some $\epsilon > 0$ and that $C^2$ is a small perturbation of $B^2$ in the sense that $\of{C^2-B^2}^{1/2}$ is trace class.
Due to the operator monotonicity of the square root we have $B \leq C$ and Theorem~\ref{thm:rootUpperBound} implies that 
\begin{align}
\Tr\brak{C - B}
\ = \ 
\Tr\brak{\of{C^2 - B^2 + B^2}^{1/2} - B}
\ \leq \
\Tr\brak{\of{C^2-B^2}^{1/2}} < \infty \, ,
\end{align}
meaning that $C - B$ also is a positive trace-class operator.
We consider the difference

\begin{align} \label{eq:rootDifference}
\Tr\brak{\of{A^2 + B^2}^{1/2} - \of{A^2 + C^2}^{1/2} + \of{C - B}}
\end{align}

and show that it is positive. 
Here Theorem~\ref{thm:rootUpperBound} ensures that both $\of{A^2 + B^2}^{1/2} - B$ and $\of{A^2 + C^2}^{1/2} - C$ are trace class, and so is their difference. 
To keep the following expressions concise we define the operators $S_{\pm}$ as 

\begin{align} \label{eq-123}
S_{\pm} 
\ := \
\of{A^2 + B^2}^{1/2} \pm \of{A^2 + C^2}^{1/2} \, .
\end{align} 

In this notation the difference \eqref{eq:rootDifference} reads
$\Tr\brak{S_{-} + \of{C - B}}$. We note at this point that $S_-$ is
trace class because $C - B$ is. Multiplying and dividing by the
positive and invertible operator $S_{+} \geq \sqrt{\eps}$ and
using that

\begin{align} \label{eq-124}
S_- S_+ 
\ = \ 
B^2 - C^2 + \brak{\of{A^2 + B^2}^{1/2} , \of{A^2 + C^2}^{1/2}}   
\end{align}

yields

\begin{align} \label{eq:TraceRootDifference}
\Tr\brak{S_{-} + \of{C - B}}
\ = \ &
\Tr\brak{S_{+}^{-1/2}\of{S_{-}S_{+} + \of{C - B}S_{+}}S_{+}^{-1/2}}
\\ \nonumber 
\ = \ &
\Tr\brak{S_{+}^{-1/2}\of{B^2 - C^2}S_{+}^{-1/2} 
+ S_{+}^{-1/2}\of{C - B}S_{+}^{1/2} + R}  \, ,
\end{align}

where the operator $R$ is given by

\begin{align} \label{eq-125}
R 
\ = \ 
S_{+}^{-1/2} \brak{\of{A^2 + B^2}^{1/2},\of{A^2 + C^2}^{1/2}} S_{+}^{-1/2} \, .
\end{align}

Recalling the definition of $S_-$, we get   

\begin{align} \label{eq-126} 
& \brak{\of{A^2 + B^2}^{1/2}, \of{A^2 + C^2}^{1/2}}    
\ = \ 
\brak{S_- + \of{A^2 + C^2}^{1/2},\of{A^2 + C^2}^{1/2}} 
\\ \nonumber & \equationindent 
\ = \ 
\brak{ S_- ,\of{A^2 + C^2}^{1/2}}
\ = \ \frac{1}{2}\brak{ S_- , 2 \of{A^2 + C^2}^{1/2}} 
\\ \nonumber & \equationindent 
\ = \ 
\frac{1}{2} \brak{ S_- \, , \, S_- + 2 \of{A^2 + C^2}^{1/2}} 
\ = \ 
\frac{1}{2} \brak{ S_- \, , \, S_+ } \, ,
\end{align} 

which is a trace-class operator because $S_-$ is.  We conclude, using
the above equation together with a direct calculation, that

\begin{align} \label{eq-127} 
\Tr\brak{R} \ = \
\frac{1}{2}\Tr\brak{ \brak{ S_-, S_+ } S_{+}^{-1} } 
\ = \ 0 \, .  
\end{align}  

Using \eqref{eq:TraceRootDifference} and \eqref{eq-127} leads us to

\begin{align}  \label{eq:traceStep}
\Tr\brak{ S_- + (C-B) } 
\ = \ 
\Tr\brak{ S_{+}^{-1/2} \of{B^2 - C^2} S_{+}^{-1/2} 
+ S_{+}^{-1/2} \of{C - B} S_{+}^{1/2}} \, .
\end{align}

Since each individual term in the above trace is trace-class we have

\begin{align}
\nonumber &
\Tr\brak{ S_{+}^{-1/2} \of{B^2 - C^2} S_{+}^{-1/2} 
+ S_{+}^{-1/2} \of{C - B} S_{+}^{1/2}} 
\\& \equationindent 
\ = \
\Tr\brak{\of{C - B} - \of{C^2 - B^2}^{1/2}S_+^{-1}\of{C^2 - B^2}^{1/2}}
\end{align}

Due to the operator monotonicity of the square root, see
\cite[Ch.~5]{Bhatia1996}, we have that $(A^2 + B^2)^{1/2} \geq B$ and
$(A^2 + C^2)^{1/2} \geq C$.  This, together with the monotonicity of
the inverse, implies that $S_+^{-1} \leq \of{B + C}^{-1}$ (recall that
$B$ and $C$ are both bounded from below by 
$\sqrt{\epsilon} > 0$), and therefore

\begin{align}
\nonumber & 
\Tr\brak{\of{C - B} - \of{C^2 - B^2}^{1/2}S_+^{-1}\of{C^2 - B^2}^{1/2}}
\\ \nonumber & \equationindent
\ \geq \ 
\Tr\brak{\of{C - B} - \of{C^2 - B^2}^{1/2}\of{C + B}^{-1}\of{C^2 - B^2}^{1/2}}
\\ & \equationindent
\ = \ 
\Tr\brak{\of{C - B} - \of{C^2 - B^2}\of{C + B}^{-1}}
\end{align}

The identity 

\begin{align} \label{eq-128}
C^2 - B^2 
\ = \
\frac{1}{2}\brak{\of{C - B} \of{C + B} + \of{C + B} \of{C - B}} 
\end{align}

together with the cyclic property of the trace then implies that 

\begin{align}
\Tr\brak{\of{C - B} - \of{C^2 - B^2}\of{C + B}^{-1}} 
\ = \
0
\end{align}

This finishes the proof of Eq.~\eqref{eq:rootLowerBound} in case that
$B^2 \geq \epsilon > 0$ and $\of{C^2 - B^2}^{1/2}$ is trace class.

Now we drop the assumption that $\of{C^2 - B^2}^{1/2}$ is trace
  class. Once again we use an approximation argument as in the proof
  of Theorem~\ref{thm:rootUpperBound}. Let $\of{\cH_L}_{L = 1}^\infty$
  be an increasing sequence of finite-dimensional subspaces. Denote
  by $Q_L$ the orthogonal projection onto $\cH_L$ and assume that $Q_L
  \leq Q_{L+1} \nearrow \bfone_\cH$ strongly, as $L \to \infty$. We
  then approximate $C$ in the strong sense by the operators

\begin{align} \label{eq-133,1}
C_L \ := \ 
\of{B^2 + Q_L \of{C^2 - B^2} Q_L}^{1/2} 
\ \geq \ B \, . 
\end{align}

Indeed, $C^2 -C_L^2 = Q_L^\perp (C^2-B^2) Q_L + Q_L (C^2-B^2) Q_L^\perp \to 0$
strongly, and by Lemma~\ref{Alem:strongConvergence} also
$C_L \to C$ strongly, as $L \to \infty$. Moreover note that, 
for any $L \in \NN$, the operator $\of{C_L^2 - B^2}^{1/2} =
(Q_L \of{C^2 - B^2} Q_L^{1/2}$ is positive, of finite rank, and 
therefore trace class, and by \eqref{eq:rootLowerBound} we have that

\begin{align} \label{eq-134}
\Tr\brak{\of{A^2 + C_L^2}^{1/2} - C_L}
\ \leq \ 
\Tr\brak{\of{A^2 + B^2}^{1/2} - B} \, .
\end{align}

Lemma~\ref{Alem:strongConvergence} and
Lemma~\ref{Alem:traceApproximation} then yield the desired inequality,

\begin{align} \label{eq-135}
\Tr\brak{\of{A^2 + C^2}^{1/2} - C}
\ \leq \ 
\sup_{L \in \NN} \Tr\brak{\of{A^2 + C_L^2}^{1/2} - C_L} 
\ \leq \ 
\Tr\brak{\of{A^2 + B^2}^{1/2} - B} \, .
\end{align}

We finally drop the remaining assumption and take $B$ and $C$ to
  be any two bounded and positive operators with $C^2 \geq B^2$. For
  every $n \in \NN$, we define $C_n := \of{C^2 + \frac{1}{n}}^{1/2}$
  and $B_n := \of{C^2 + \frac{1}{n}}^{1/2}$. Of course, $C_n^2 \to
  C^2$ and $B_n^2 \to B^2$, as $n \to \infty$, in norm, and
  Lemma~\ref{Alem:strongConvergence} implies that $C_n \to C$ and $B_n
  \to B$, as $n \to \infty$, strongly.  Clearly, $C^2_n \geq
  B^2_n \geq \frac{1}{n} > 0$ and

\begin{align}
\Tr\brak{\of{A^2 + C_n^2}^{1/2} - C_n}
\ \leq \ 
\Tr\brak{\of{A^2 + B_n^2}^{1/2} - B_n} \, ,
\end{align} 

by \eqref{eq-135}. Due to Lemma \ref{lem:operatorMon} we have

\begin{align}
\Tr\brak{\of{A^2 + B_n^2}^{1/2} - B_n}
\ \leq \
\Tr\brak{\of{A^2 + B^2}^{1/2} - B} \, ,
\end{align}

for every $n \in \NN$. Taking the limit $n \to \infty$ and
invoking Lemma~\ref{Alem:strongConvergence} and
Lemma~\ref{Alem:traceApproximation} then yields

\begin{align}
\Tr\brak{\of{A^2 + C^2}^{1/2} - C}
\ \leq \ 
\sup_{n \in \NN} \Tr\brak{\of{A^2 + C_n^2}^{1/2} - C_n} 
\ \leq \
\Tr\brak{\of{A^2 + B^2}^{1/2} - B} \, .
\end{align}

\end{proof}
\end{thm}

\section{Non-Convexity of the Full Energy Functional}
\label{sec:NonConvexity}

In this section we provide the proof of the non-convexity of the full
energy functional $\cE_{g,\Lambda}$. Recall that $\cE_{g,\Lambda}$ is
defined on the space 
$\HS\of{\fh_\ph^\Lambda} \oplus \fh_\ph^\Lambda$ and is
given by

\begin{align} \label{eq-136}
\cE_{g,\Lambda}\of{z,\eta} 
\ := \ & 
\frac{1}{8}\sum_{\nu = 1}^3 
\Tr\brak{k_\nu^2 z^2 \of{1 + z}^{-1} - k_\nu z k_\nu z \of{1 + z}^{-1}}
\\ & 
+ \frac{1}{2}\sum_{\nu = 1}^3 
\sca{G_\nu + k_\nu \eta}{\of{1 + z}^{-1}\of{G_\nu + k_\nu \eta}}
\\ &
+ \frac{1}{4} \Tr\brak{\abs{k} z^2 \of{1 + z}^{-1}} 
+ \sca{\eta}{\abs{k}\eta} \, .
\end{align}

In particular, we identify the term 

\begin{align} \label{eq-137}
\frac{1}{8}\sum_{\nu = 1}^3 
\Tr\brak{\abs{k}^2 z^2\of{1 + z}^{-1} - k_\nu z k_\nu z\of{1 + z}^{-1}}
\end{align}

to be the source of the lack of convexity
of $\cE_\BHF^{g,\Lambda}$. 

\begin{thm} \label{thm:nonConvexity}
For any fixed $\eta \in \fh_\Lambda$ the functional
$\cE_{g,\Lambda}\of{z,\eta}$ is not convex in $z$ on
$\HS_{0}\of{\fh_\Lambda}$.

\begin{proof}
We provide a counterexample for the convexity in the space of rank-two
operators. Let $\psi, \vphi \in \fh_\Lambda$ be two normalized
and mutual orthogonal vectors, that is 
$\norm{\psi}_{\fh_\Lambda} = \norm{\vphi}_{\fh_\Lambda} = 1$ 
and $\sca{\vphi}{\psi} = 0$. Additionally suppose that $\vphi$ and 
$\psi$ are both orthogonal to $G_\nu + k_\nu \eta$, for all $\nu = 1,2,3$, 
as well as $\sca{k_\nu \psi}{\vphi} \neq 0$, for some $\nu$.  
For $t,s > 0$ the operator $z_{t,s} \in \HS_0\of{\fh_\Lambda}$ 
is defined as

\begin{align} \label{eq-138}
z_{t,s} \ := \
t P_\psi + s P_\vphi \, ,
\end{align}

where $P_f$ is the orthogonal projection onto the span
of $f$. We demonstrate that the two-parameter function 
$\of{t,s} \mapsto \cE_{g,\Lambda}\of{z_{t,s},\eta}$ is not
convex in $\of{t,s}$ on $(\RR_{0}^+)^2$. This presents a counterexample to
the convexity of $\cE_{g,\Lambda}$.  

To this end we write
$\cE_{g,\Lambda}\of{z,\eta} = \frac{1}{4}
\cE^{\prime}_{\Lambda}\of{z} +
\cE^{\prime\prime}_{g,\Lambda}\of{z,\eta}$, where

\begin{align} \label{eq-139}
& \cE^{\prime}_{\Lambda}\of{z} 
\ = \
\Tr\brak{\abs{k} z^2\of{1 + z}^{-1}} 
+ \frac{1}{2}\sum_{\nu = 1}^3 
\Tr\brak{\abs{k}^2 z^2\of{1 + z}^{-1} - k_\nu z k_\nu z\of{1 + z}^{-1}} \, ,
\\ \label{eq-140}
& \cE^{\prime\prime}_{g,\Lambda}\of{z,\eta}
\ = \
\sca{\eta}{\abs{k}\eta} + \sum_{\nu = 1}^3 
\sca{G_\nu + k_\nu \eta}{\of{1 + z}^{-1}\of{G_\nu + k_\nu \eta}} \, ,
\end{align}

and note that the term $\cE^{\prime}_{\Lambda}\of{z}$ does not depend
on the coupling constant $g$. The specific choice of the projections
$P_\psi$ and $P_\vphi$ ensures that the term
$\cE^{\prime\prime}_{g,\Lambda}\of{z_{t,s},\eta}$ is constant in $t$
and $s$, as $(1+z)^{-1} (G_\nu + k_\nu \eta) = (G_\nu + k_\nu \eta)$.  
Hence, it suffices to analyze the
convexity of $\cE_{\Lambda}^\prime\of{z_{t,s}}$ in $\of{t,s}$.  

The expression $\cE^{\prime}_{\Lambda}\of{z_{t,s}}$ is given by

\begin{align} \label{eq-141}
\cE^{\prime}_{\Lambda}\of{z_{t,s}}
\ = \ &
\frac{t^2}{1 + t} \sca{\of{\abs{k} + \frac{1}{2}\abs{k}^2 }\psi}{\psi} 
+ \frac{s^2}{1 + s} \sca{\of{\abs{k} + \frac{1}{2}\abs{k}^2 }\vphi}{\vphi} 
\\ &- 
\frac{t}{2(1 + t)} 
\sum_{\nu = 1}^3 \sca{k_\nu \psi}{z_{t,s} k_\nu \psi}
- \frac{s}{2(1 + s)} 
\sum_{\nu = 1}^3 \sca{k_\nu \vphi}{z_{t,s} k_\nu \vphi} \, ,
\end{align}

due to

\begin{align} \label{eq-142}
& \Tr\brak{\of{\abs{k} + \frac{1}{2}\abs{k}^2} z_{t,s}^2\of{1 + z_{t,s}}^{-1}}
\\ \nonumber 
& \equationindent \ = \
\frac{t^2}{1 + t} 
\sca{\of{\abs{k} + \frac{1}{2}\abs{k}^2 }\psi}{\psi} 
+ \frac{s^2}{1 + s} 
\sca{\of{\abs{k} + \frac{1}{2}\abs{k}^2 }\vphi}{\vphi} \, ,
\end{align} 

and 

\begin{align} \label{eq-143}
& \frac{1}{2} \sum_{\nu = 1}^3 
\Tr\brak{k_\nu z_{t,s} k_\nu z_{t,s}\of{1 + z_{t,s}}^{-1}} 
\\ \nonumber 
& \equationindent \ = \
\frac{t}{2(1 + t)} 
\sum_{\nu = 1}^3 \sca{k_\nu \psi}{z_{t,s} k_\nu \psi}
+ \frac{s}{2(1 + s)} 
\sum_{\nu = 1}^3 \sca{k_\nu \vphi}{z_{t,s} k_\nu \vphi} \, .
\end{align}

Furthermore the two identities  

\begin{align} \label{eq-144}
\sca{k_\nu \vphi}{z_{t,s} k_\nu \vphi} 
\ = \ & 
s \abs{\sca{k_\nu \vphi}{\vphi}}^2 + t \abs{\sca{k_\nu \psi}{\vphi}}^2 
\\ \label{eq-145}
\sca{k_\nu \psi}{z_{t,s} k_\nu \psi} 
\ = \ &
t \abs{\sca{k_\nu \psi}{\psi}}^2 +s \abs{\sca{k_\nu \psi}{\vphi}}^2 
\end{align}

allow for $\cE^{\prime}_{\Lambda}\of{z_{t,s}}$ to be fully expressed
through the vectors $\vphi$ and $\psi$ as

\begin{align} \label{eq-146}
\frac{t}{2(1 + t)} \sum_{\nu = 1}^3 & \sca{k_\nu \psi}{z_{t,s} k_\nu \psi}
+ \frac{s}{2(1 + s)} \sum_{\nu = 1}^3 \sca{k_\nu \vphi}{z_{t,s} k_\nu \vphi} 
\nonumber \\ 
\ = \ &
\frac{t}{2(1 + t)} 
\sum_{\nu = 1}^3 \abs{\sca{k_\nu \psi}{\psi}}^2
+ \frac{s}{2(1 + s)} \sum_{\nu = 1}^3 \abs{\sca{k_\nu \vphi}{\vphi}}^2
\nonumber \\ &
+ \of{ \frac{ts}{2(1 + t)} + \frac{ts}{2(1 + s)} } 
\sum_{\nu = 1}^3\abs{\sca{k_\nu \psi}{\vphi}}^2 \, .
\end{align}

Inserting these identities, $\cE^{\prime}_{\Lambda}\of{z_{t,s}}$ 
assumes the explicit form 

\begin{align} \label{eq-147}
\cE\of{z_{t,s}} 
\ = \ &
\frac{t^2}{1 + t} 
\of{ \sca{\of{\abs{k} + \frac{1}{2}\abs{k}^2 }\psi}{\psi} 
- \frac{1}{2} \sum_{\nu = 1}^3 \abs{\sca{k_\nu \psi}{\psi}}^2}
\nonumber \\ &
+ \frac{s^2}{1 + s} 
\of{ \sca{\of{\abs{k} + \frac{1}{2}\abs{k}^2 }\vphi}{\vphi} 
- \frac{1}{2} \sum_{\nu = 1}^3 \abs{\sca{k_\nu \vphi}{\vphi}}^2}
\nonumber \\ &
- \frac{1}{2} \of{\frac{ts}{1 + t} + \frac{ts}{1 + s}} 
\sum_{\nu = 1}^3\abs{\sca{k_\nu \psi}{\vphi}}^2 \, .
\end{align}

Hence, there are positive constants $a,b,c > 0$ depending only on 
$\psi$ and $\vphi$ such that 

\begin{align} \label{eq-148}
\cE\of{z_{t,s}}  
\ = \
 a\frac{t^2}{1 + t} + b\frac{s^2}{1 + s} - c\of{\frac{ts}{1 + t} + \frac{ts}{1 + s}} \, .
\end{align} 

We observe that the constants $a$ and $b$ are strictly greater than
$c$. More precisely, it is

\begin{align} \label{eq:coef}
a,b \ \geq \ \frac{1}{\Lambda} + c \, .
\end{align}

The reason for this is that 

\begin{align} \label{eq-149}
\sca{k_\nu^2\psi}{\psi}
\ \geq \ 
\sca{k_\nu\psi}{\of{P_\psi + P_\vphi}k_\nu \psi}
\ = \ 
\abs{\sca{k_\nu \psi}{\psi}}^2 + \abs{\sca{k_\nu \psi}{\vphi}}^2 \, ,
\end{align}

which allows for the estimate

\begin{align} \label{eq-150}
a \ = \ & 
\sca{\of{\abs{k} + \frac{1}{2}\abs{k}^2 }\psi}{\psi} 
- \frac{1}{2} \sum_{\nu = 1}^3 \abs{\sca{k_\nu \psi}{\psi}}^2 
\nonumber \\ 
\ = \ &
\sca{\abs{k}\psi}{\psi} 
+ \frac{1}{2} \sum_{\nu = 1}^3 
\of{ \sca{k_\nu^2\psi}{\psi} - \abs{\sca{k_\nu \psi}{\psi}}^2 }
\nonumber \\ 
\ \geq \ &
\sca{\abs{k}\psi}{\psi}  
+ \frac{1}{2} \sum_{\nu = 1}^3 \abs{\sca{k_\nu \psi}{\vphi}}^2
\nonumber \\ 
\ \geq \ &
\Lambda^{-1} + \frac{1}{2} \sum_{\nu = 1}^3 \abs{\sca{k_\nu \psi}{\vphi}}^2
\ = \ 
\Lambda^{-1} + c \, .
\end{align}

The analogous estimate for $b$ is proven the same way.

The convexity of $\of{t,s} \mapsto \cE^{\prime}_{\Lambda}\of{z_{t,s}}$
is now examined via its Hessian. A straightforward
calculation shows that

\begin{align} \label{eq-151}
\nabla_{\of{t,s}} \cE^{\prime}_{\Lambda}\of{z_{t,s}} 
\ = \ 
\begin{pmatrix}
a\frac{t^2 + 2t}{\of{1 + t}^2} 
- c \frac{s}{\of{1 + t}^2} - c \frac{s}{1 + s}
\\[1ex]
b\frac{s^2 + 2s}{\of{1 + s}^2} 
- c \frac{t}{1 + t} - c \frac{t}{\of{1 + s}^2}
\end{pmatrix} \, ,
\end{align}

and 

\begin{align} \label{eq-152}
\cH_{\of{t,s}} \cE^{\prime}_{\Lambda}\of{z_{t,s}}
\ = \
\begin{pmatrix}
\frac{2a + 2cs}{\of{1 + t}^3} 
& - c \of{\frac{1}{\of{1 + t}^2} + \frac{1}{\of{1 + s}^2}} 
\\[1ex]
- c \of{\frac{1}{\of{1 + t}^2} + \frac{1}{\of{1 + s}^2}} 
& \frac{2b + 2ct}{\of{1 + s}^3}
\end{pmatrix} \, .
\end{align}

The trace of this Hessian equals 

\begin{align} \label{eq-153}
\Tr\brak{\mathcal{H}_{\of{t,s}} \cE^{\prime}_{\Lambda}\of{z_{t,s}}} 
\ = \ 
\frac{2a + 2cs}{\of{1 + t}^3} + \frac{2b + 2ct}{\of{1 + s}^3} 
\ > \ 0 \, ,
\end{align}

while the determinant of 
$\mathcal{H}_{\of{t,s}} \cE^{\prime}_{\Lambda}\of{z_{t,s}}$ equals

\begin{align} \label{eq:determinant}
\det \mathcal{H}_{\of{t,s}} \cE^{\prime}_{\Lambda}\of{z_{t,s}}
& \ = \
\frac{\of{2a + 2cs}\of{2b + 2ct}}{\of{1 + t}^3\of{1 + s}^3} 
- c^2 \of{\frac{1}{\of{1 + t}^2} + \frac{1}{\of{1 + s}^2}}^2 \, .
\end{align}

As the trace of the Hessian is always positive,
$\cE^{\prime}_{\Lambda}\of{z_{t,s}}$ is convex if and only if its determinant 
$\det \mathcal{H}_{\of{t,s}} \cE^{\prime}_{\Lambda}\of{z_{t,s}}$ is
positive. According to \eqref{eq:determinant} this is the case if, and
only if,

\begin{align} \label{eq:detPositivity}
\frac{\of{2a + 2cs}\of{2b + 2ct}}{\of{1 + t}^3\of{1 + s}^3} 
- c^2 \of{\frac{1}{\of{1 + t}^2} + \frac{1}{\of{1 + s}^2}}^2 > 0 \, .
\end{align}

The above inequality does not, however, hold on the entirety of 
$(\RR_0^+)^2$. To see this we recall \eqref{eq:coef} and estimate

\begin{align} \label{eq:detEstimate}
& \frac{\of{2a + 2cs}\of{2b + 2ct}}{\of{1 + t}^3\of{1 + s}^3} 
- c^2 \of{\frac{1}{\of{1 + t}^2} + \frac{1}{\of{1 + s}^2}}^2 
\nonumber \\ 
& \equationindent \ \leq \
\frac{4ab}{\of{1 + t}^2 \of{1 + s}^2} 
- c^2 \of{\frac{1}{\of{1 + t}^2} + \frac{1}{\of{1 + s}^2}}^2 
\nonumber \\ 
& \equationindent \ = \
\frac{4ab - 2c^2}{\of{1 + t}^2 \of{1 + s}^2} 
- \frac{c^2}{\of{1 + t}^4} - \frac{c^2}{\of{1 + s}^4} \, .
\end{align}

The upper bound in \ref{eq:detEstimate}, and thus also 
$\det \mathcal{H}_{\of{t,s}} \cE^{\prime}_{\Lambda}\of{z_{t,s}}$, is
strictly negative for $t \gg 1$ and $0 < s \ll 1$, hence disproving
the convexity of $\cE^{\prime}_{\Lambda}\of{z_{t,s}}$ on
$(\RR_0^*)^2$. We conclude that 
$\cE_{g,\Lambda}\of{z, \eta}$ is not convex in $z$ on 
$\HS_{0}\of{\fh_\Lambda}$.
\end{proof}
\end{thm}

As convexity of a functional is not necessarily preserved under
reparametrization, the non-convexity of $\cE_{g,\Lambda}\of{z,\eta}$
could be the result of an unfavorable choice of reparametrization
between Proposition~\ref{prop:Intro_OriginalProblem} and
\ref{prop:Intro_FullEnergyBounds}. The reason that we regard the
present parametrization as suitable is the fact that the contribution
of the interaction in $\cE_{g,\Lambda}\of{z,\eta}$ is indeed convex,
as we have seen in Section~\ref{sec:Existence}.

\section{Appendix}
\label{sec:Appendix}

\begin{thm} \label{thm:strictConvexity}
Let $A$ and $B$ be two positive and bounded operators with bounded
inverses. Assume $A \neq B$ and let $\lambda \in \of{0,1}$. Then

\begin{align} \label{eq-A.01}
0 \ \leq \ 
\lambda A^{-1} + \of{1 - \lambda} B^{-1} 
- \brak{\lambda A + \of{1 - \lambda} B}^{-1}
\ \neq \ 0 \, .
\end{align}

\begin{proof}
The statement can be proven by carefully following a standard proof
for the convexity of the inverse. For this purpose, we define the
operator $C = A^{-1/2} B A^{-1/2}$ and observe that

\begin{align} \label{eq-A.02}
\brak{\lambda A + \of{1 - \lambda} B}^{-1}
\ = \
A^{-1/2} \brak{\lambda + \of{1 - \lambda} C}^{-1} A^{-1/2} \, .
\end{align}

For all $x > 0$ with $x \neq 1$ it is 
$\brak{\lambda + \of{1 - \lambda} x}^{-1} < 
\lambda + \of{1 - \lambda} x^{-1}$. 
Due to the fact that $A \neq B$, it follows that $C \neq 1$. 
Therefore, we obtain that

\begin{align} \label{eq-A.03}
\brak{\lambda + \of{1 - \lambda} C}^{-1} 
\ \leq \
\lambda +  \of{1 - \lambda} C^{-1}
\end{align}

and 

\begin{align} \label{eq-A.04}
\brak{\lambda + \of{1 - \lambda} C}^{-1} 
\ \neq \
\lambda \of{1 - \lambda} C^{-1} \, .
\end{align}

Altogether we thus have

\begin{align} \label{eq-A.05}
\brak{\lambda A + \of{1 - \lambda} B}^{-1}
\ = \ &
A^{-1/2} \brak{\lambda + \of{1 - \lambda} C}^{-1} A^{-1/2}
\\ \nonumber 
\ \leq \ &
A^{-1/2} \brak{\lambda +  \of{1 - \lambda} C^{-1}} A^{-1/2}
\ = \
\lambda A^{-1} + \of{1 - \lambda} B^{-1} \, .
\end{align}

This, together with 

\begin{align} \label{eq-A.06}
\brak{\lambda A + \of{1 - \lambda} B}^{-1}
\ \neq \ 
\lambda A^{-1} + \of{1 - \lambda} B^{-1} \, ,
\end{align}

proves the statement.
\end{proof}
\end{thm}

\begin{thm} \label{thm:minimizer}
Let $U$ be a weakly closed subset of a Hilbert space and let 
$f: U \to \RR$ be a weakly lower semi-continuous and coercive 
function. Then there is a $u_0 \in U$ such that

\begin{align} \label{eq-A.07}
\inf_{u \in U} f\of{u} 
\ = \ 
f\of{u_0} \, .
\end{align}

\begin{proof}
For the convenience of the reader, we recall that $f$ is coercive
on $U$ if for every sequence $(x_n)_{n \in \NN} \in U^\NN$ such
that $\|x_n \| \to \infty$, as $n$ tends to infinity, it follows that
$f(x_n)$ diverges to infinity. Moreover, $f$ is weakly semi-continuous
if $\liminf_{n \to \infty}\{ f(x_n) \} \geq f(x_0)$, for every
weakly convergent sequence $(x_n)_{n \in \NN} \in U^\NN$ with
$x_0 := w-\lim_{n \to \infty} \{x_n\}$.
  
Let $\of{u_n}_{n = 1}^\infty \subset U$ be a minimizing sequence of
$f$ in $U$. As $f$ is coercive this sequence must be bounded. Due to
the Banach-Alaoglu Theorem, there must be an accumulation point $u_0$
of the sequence in the weak$^*$ topology. The latter coincides,
however, with the weak topology on Hilbert spaces, and $u_0$ is a weak
accumulation point, in fact.  Without loss of generality we may assume
$u_n \too{w} u_0$. Finally the weak lower semi-continuity of $f$
yields

\begin{align} \label{eq-A.08}
\inf_{u \in U} f\of{u} 
\ = \ 
\liminf_{n \to \infty }f\of{u_n} 
\ \geq \
f\of{u_0} 
\ \geq \ 
\inf_{u \in U} f\of{u} 
\end{align}

and hence $\inf_{u \in U} f\of{u} = f\of{u_0}$.
\end{proof}
\end{thm}

\subsection*{Acknowledgement}
Research supported by CONACYT, FORDECYT–PRONACES 429825/2020 (proyecto
apoyado por el FORDECYT–PRONACES, PRONACES/429825), recently renamed
project CF-2019 / 429825. This work was supported by the project
PAPIIT–DGAPA–UNAM IN114925. M.~B.\ is a Fellow of the Sistema
Nacional de Investigadores (SNI). Furthermore,
support from DFG Grant Nr.~BA~1477/17-1, Project Nr.~547264922,
of the German Science Foundation is gratefully acknowledged.
The authors thank Matthias Herdzik for useful discussion.

\end{document}